\documentclass[a4paper,11pt, english,notitlepage,sumlimits]{article}

\usepackage[normalem]{ulem}

\usepackage[utf8x]{inputenc}
\usepackage[a4paper,top=3cm,bottom=3cm,left=3cm,right=3cm,marginparwidth=1.75cm]{geometry}
\usepackage{dsfont}
\usepackage{mathtools} 
\usepackage{subfig}
\usepackage{multirow}
\usepackage{xcolor}
\usepackage{enumitem}
\usepackage{amssymb}
\usepackage{fancyhdr}
\usepackage{lmodern}
\usepackage[T1]{fontenc}
\usepackage{caption}
\usepackage{float}
\usepackage{cancel} 
\usepackage{graphics} 
\usepackage{amsthm}
\usepackage{revsymb}
\usepackage{authblk}
\usepackage{amsmath}
\usepackage[most]{tcolorbox}
\usepackage{braket}

\usepackage[colorlinks=true,allcolors=black]{hyperref} 

\usepackage[style=ieee]{biblatex}
\usepackage{comment}

\newtheorem{theorem}{Theorem}
\newtheorem{lemma}[theorem]{Lemma}
\newtheorem{proposition}[theorem]{Proposition}

\newcommand{\norm}[1]{\left\lVert#1\right\rVert}

\newlength{\blank}
\newenvironment{myproof}[1][{\hspace{-\blank}}]{{\medskip\noindent\textit{Proof~{#1}.~{}}}}{\hfill$\qedsymbol$}

\newcommand{\tr}{\operatorname{Tr}}
\newcommand{\Tr}{\operatorname{Tr}}
\newcommand{\1}{\openone}

\newcommand{\cA}{\mathcal{A}}
\newcommand{\cB}{\mathcal{B}}

\newcommand{\cE}{\mathcal{E}}

\newcommand{\ox}{\otimes}

\newcommand{\proj}[1]{\ket{#1}\!\bra{#1}}

\newcommand{\aw}[1]{\textcolor{blue}{#1}}

\begin{document}

\title{Robust 
Quantum Key Distribution\protect\\ Arbitrarily Close to Local Correlations}
\author[1,$\bullet$]{Hari Krishnan SV}
\author[1,2,3,$\circ$]{Andreas Winter}

\affil[1]{\small Grup d'Informaci\'o Qu\`antica, Departament de F\'isica,\protect\\ Universitat Aut\`onoma de Barcelona, 08193 Bellaterra (BCN), Spain}
\affil[$\bullet$]{Email: {\tt harikrishnan.sasidharan@uab.cat}\vspace{2mm}}
\affil[2]{\small ICREA--Instituci\'o Catalana de la Recerca i Estudis Avan{\c{c}}ats,\protect\\ Pg.~Llu\'is Companys, 23, 08010 Barcelona, Spain}
\affil[3]{\small Department Mathematik/Informatik--Abteilung Informatik,\protect\\ Universit\"at zu K\"oln, Albertus-Magnus-Platz, 50923 K\"oln, Germany}
\affil[$\circ$]{Email: {\tt andreas.winter@uni-koeln.de}}

\date{\small (14 November 2026)}

\maketitle

\begin{abstract}
Recently, Wooltorton \emph{et al.} [\href{https://doi.org/10.1103/PhysRevLett.132.210802}{Phys. Rev. Lett. 132, 210802 (2024)}] 
and Farkas [\href{ https://doi.org/10.1103/PhysRevLett.132.210803}{Phys. Rev. Lett. 132, 210803 (2024)}] 
have exhibited the mismatch between Bell inequality violations and their cryptographic application in device-independent quantum key distribution, by showing that arbitrarily close to the set of local behaviours there exist quantum correlations guaranteeing a constant rate of secret key.
While these results require correlations attaining the maximum quantum value of a suitable Bell observable (aka the Tsirelson bound) and rely on a kind of ideal self-testing of a maximally entangled state and associated Bell measurement, here we show that the effect is robust: for every one of the Bell inequalities considered by Wooltorton \emph{et al.}, a constant rate of secret key ensues if the observed Bell violation is sufficiently close to the respective Tsirelson bound. 
For these and also the Bell inequalities of Farkas, we furthermore present numerical results based on semidefinite relaxations of the minimum conditional von Neumann entropy  consistent with a certain Bell violation, which demonstrate that small but nonzero key rates can be guaranteed (in principle) by practical and efficient means. 
\end{abstract}

\section{Introduction}
\label{sec:intro}
Bell's theorem \cite{Bell_1964} is one of the most ground-breaking ideas in physics in recent times. It states that there are quantum correlations that cannot be produced by systems with an underlying local hidden variable model. This nonlocality is one of the foundations of quantum information theory and is at the heart of the security of device-independent quantum key distribution (DIQKD) protocols \cite{Pironio_2009}. In DIQKD, two parties (traditionally called Alice and Bob) aim to generate a secret key using devices that are not characterised. This means that neither the state, the measurements, nor indeed the quantum systems involved are under their control. However, the parties can guarantee nonlocality by checking for a sufficiently large Bell inequality violation, and the magnitude of the observed violation can be used to bound the information that an eavesdropper may hold about their outcomes, without any assumption on the inner workings of the devices.

The origins of this approach can be traced back to Bennett and Brassard's famous BB84 protocol \cite{bennett1984proceedings}, in which Alice and Bob generate a shared secret key by preparing, sending, and measuring qubits (e.g. the polarization states of photons). Ekert proposed the entanglement-based QKD protocol \cite{Ekert_1991}, where Alice and Bob share a Bell state and check for a violation of the Clauser-Horne-Shimony-Holt (CHSH) inequality \cite{CHSH}. For a sufficiently large violation, they can certify the presence of nonclassical correlations, which limits the information an eavesdropper (Eve) can have obtained. Nevertheless, these protocols are ultimately device-dependent.  
%
The first truly device-independent protocol of key generation was given by Barrett, Hardy and Kent for a general adversary limited only by the no-signalling principle \cite{Barrett_2005}. This protocol had very low key rate and lacked noise tolerance. A security analysis for a robust protocol with a quantum eavesdropper performing memoryless collective (also known as the independent and identically distributed (i.i.d.)) attacks was given in \cite{Pironio_2009}. 
Security proofs for a general adversary restricted only by no-signalling were given in \cite{Vazirani_2014, Miller_2016,9062494,Masanes2014}. These results, however, had the disadvantage of again resulting in very low asymptotic key rates. A fully general security analysis of DIQKD assuming the validity of quantum mechanics also for the adversary and with asymptotic key rates comparable to the i.i.d. scenario was given in \cite{Arnon_Friedman_2018} using the Entropy Accumulation Theorem (EAT) \cite{Dupuis2019,Dupuis_2020}. 


An interesting fundamental question to ask is how much nonlocality is needed in order to obtain a non-zero key rate, or whether there is in fact a relation between the degree of nonlocality observed and the capacity to guarantee DI secret key. Recently, it was shown that nonlocality is actually not sufficient for DIQKD in a class of standard protocols \cite{Farkas_2021}, where the Bell value may be more than the classical limit, yet no secret key can be extracted, since Eve can reproduce the observed correlations with a so-called convex-combination attack.
In the other direction, \cite{Wooltorton_2024,Farkas_2024} have shown that certain nonlocal behaviours arbitrarily close to being local can nonetheless guarantee maximal key rates. This was done by self-testing a maximally entangled state and certain local measurements at the point of maximal violation of two families of Bell inequalities. 
However, in practice there is no way to be certain of a maximal violation: all that finite (and possibly noisy) statistics can give is a confidence interval, implying a lower bound on the true Bell violation of the behaviour concerned. This prompts the next question of whether the secret key rate is bounded away from zero if the Bell violation is sufficiently close to maximal. 

The notion of bounding the strategy achieving a suboptimal violation of a Bell inequality is referred to as robust self-testing. This was first explored in \cite{mmmo2006} where the authors gave a device-independent robust version of the Mayers-Yao self-test \cite{Mayers2004SelfTesting}. 
Robust self-testing of the CHSH inequality was presented in \cite{PhysRevA.80.062327} and a simplified protocol for both Mayers-Yao and CHSH self-tests was given in \cite{McKague_2012}. Since then, there has been a lot of progress in robust self-testing of states and measurements \cite{PhysRevA.98.042336,PhysRevA.91.052111, PhysRevA.90.042339,PhysRevLett.121.180505}, with the current best robust bounds for two-qubit states given in \cite{PhysRevLett.117.070402}.  

Here, we approach the question of robust DI key rates rigorously for the family of Bell inequalities in \cite{Wooltorton_2024}. We propose a robust self-test, where for a violation that is within an $\epsilon>0$ from the Tsirelson bound, the state and the measurements are bounded by a distance that is a function of $\epsilon$, from the essentially unique strategy achieving the maximal violation. This subsequently also guarantees a nonzero key rate. 
We also give numerical bounds for the key rates using the methods in \cite{Brown_2024}, both for the family of protocols from \cite{Wooltorton_2024} and the one from \cite{Farkas_2024}.

\section{DIQKD protocols}
\label{sec:prelim}
We start with an overview of DIQKD protocols. Consider two players, Alice and Bob, who share a state, and can choose between settings $x$ and $y$ to feed into their local devices, producing outputs $a$ and $b$, respectively. In DIQKD, their devices are uncharacterized and are only restricted by no-signaling (meaning that if Alice and Bob use them in a spacelike manner, they must act independently on whatever state was distributed before). 
In the quantum scenario, the eavesdropper, Eve, is assumed (without loss of generality) to hold a purification of Alice and Bob's state, and is thus potentially correlated (and even entangled) with them: $\ket{\psi}^{\cA\cB\cE}$. The protocol consists of two different types of rounds, which are chosen randomly among all $n\gg 1$ rounds: Bell test rounds and key generation rounds. In each round, Alice applies a POVM $A^x = (A^x_a)$ from a finite set $x\in X$, and Bob a POVM $B^y = (B^y_b)$ from a finite set $y\in Y$. 

In a Bell test round, Alice and Bob certify the nonlocality of their state by testing for the violation of a Bell inequality, whose maximum value over local hidden-variable strategies is $\beta$ (the Bell bound), while the maximum/supremum over quantum strategies (the Tsirelson bound) has to be $\tau > \beta$, and they require the observation of a sufficiently large value $\eta > \beta$. In these rounds, statistics are also gathered on the joint distribution of Alice and Bob's observations. 

In the key generation rounds, Alice and Bob obtain a raw key by performing a measurement on their respective parts of the shared state, corresponding to specific settings $x^*\in X$ and $y^*\in Y$ of the Bell scenario. The figure of merit is the key rate $R$, defined as the number of private shared key bits generated by Alice and Bob per round of the protocol (for a given security parameter). In the most general setting, Eve can apply a coherent attack, where she uses a different attack in each round of the protocol potentially keeping classical and quantum memory from previous rounds. For such an adversary, the best bound on the key rate was given in \cite{Arnon_Friedman_2018}, where the authors obtained a lower bound on the smooth min-entropy between Alice and Eve. In the asymptotic limit of $n\rightarrow\infty$, this bound reduces to the i.i.d. scenario, where each round is independent of the previous round and Eve applies the same independent attack in each round. In this limit, the key rate guaranteed by the protocol is given by the familiar expression \cite{Devetak_2005}
\begin{equation}
  \label{dw}
  R \geq \inf\, \bigl\{ H(A|\cE,X=x^*) - H(A|B,X=x^*,Y=y^*) \bigr\},
\end{equation}
where the entropies are evaluated on the classical-quantum state 
\[
  \Omega^{XYAB\cE} 
    = \sum_{xy} p(xy) \proj{x}^X \ox \proj{y}^Y 
       \ox \sum_{ab} \proj{a}^A \ox \proj{b}^B 
       \ox \tr_{\cA\cB} \psi(A^x_a\ox B^y_b\ox\1_{\cE}), 
\]
and the infimum is over all states $\psi^{\cA\cB\cE}$ of the joint system held by Alice, Bob, and Eve, and over all measurements done by Alice and Bob that are consistent with the estimated Bell score $\eta$ and the estimated probabilities of $A$ and $B$ given $X=x^*,\,Y=y^*$.

In the present paper, we study the family of Bell inequalities first introduced in \cite{Le_2023} and recently studied in \cite{Wooltorton_2024}. In those, Alice and Bob receive inputs $x,y\in\{0,1\}$ and they produce outputs $a,b\in\{0,1\}$. A quantum strategy in a finite-dimensional Hilbert space $\cA\otimes\cB$ is then given by the tuple $S = (\rho^{\cA\cB}, A^x, B^y)$ consisting of a density matrix $\rho^{\cA\cB}$ and observables $A^x=A^x_0 - A^x_1$ and $B^y=B^y_0 - B^y_1$, where $(A^x_a)_a$ and $(B^y_b)_b$ are POVMs on $\cA$ and $\cB$, respectively, which without loss of generality can be assumed to be projective (via the Naimark extension theorem, by embedding the local Hilbert spaces into larger ones). 
Let $\theta, \phi, \omega \in \mathbb{R}$ and define the family of Bell expressions from \cite{Le_2023},
\begin{equation}\begin{split}
  \label{bell function}
  T_{\theta,\phi,\omega} 
    &:= \cos(\theta + \phi)\cos(\theta + \omega) A^0 \ox \bigl(\cos\omega\,B^0 - \cos\phi\,B^1\bigr) \\
    &\phantom{=}
       +\cos\phi \cos\omega\, A^1 \ox \bigl(-\cos(\theta + \omega)\, B^0 + \cos(\theta + \phi)\, B^1\bigr).
\end{split}\end{equation}
Setting $\theta = -\frac{\pi}{2}$, $\phi = \frac{3\pi}{4}$, and $\omega = \frac{\pi}{4}$ recovers the standard CHSH correlator, up to a scaling factor of $\tfrac{1}{2\sqrt{2}}$. 

We denote the local (Bell) and quantum (Tsirelson) bounds
\begin{align}
  \beta(\theta,\phi,\omega) &= \max_{u=\pm 1} |\cos(\theta+\omega)(\cos\omega)(\cos(\theta+\phi)+u\cos\phi)| \nonumber\\
  &\phantom{=\max} + |\cos(\theta+\phi)(\cos\phi)(\cos(\theta+\omega)-u\cos\omega)| ,\text{ and} \\
  \tau(\theta,\phi,\omega)  &= \left| \sin\theta \sin(\omega-\phi) \sin(\theta+\omega+\phi)\right|,
\end{align}
respectively, cf.~\cite{Wooltorton_2024}. 
In \cite[Prop.~8~{\&}~Thm.~19]{Le_2023} and \cite{Wooltorton_2024} it is proved that whenever 
\begin{equation}
\label{eq: condition for quantum greater than classical}  
    \cos(\theta + \phi)\cos\phi \; \cos(\theta + \omega)\cos\omega < 0,
\end{equation}
then $\tau(\theta,\phi,\omega) > \beta(\theta,\phi,\omega)$, and the attainment of the maximal quantum value, i.e. $\langle T_{\theta,\phi,\omega} \rangle = \tau(\theta,\phi,\omega)$, self-tests the following quantum strategy up to dummy gauge systems and local isometries:
\begin{align}
\label{Eq:optimal quantum strategy}
\psi &= \ket{\Phi_0}\!\bra{\Phi_0},
\quad
\text{where }
\ket{\Phi_0} = \frac{1}{\sqrt{2}}(\ket{00} +  \ket{11}) \nonumber\\
A^{Q0} &= \sigma_Z,
           \quad\qquad\qquad\qquad\qquad\, 
          A^{Q1} = (\cos\theta)\,\sigma_Z + (\sin\theta)\,\sigma_X,  \\
B^{Q0} &=(\sin\phi)\,\sigma_Z + (\cos\phi)\,\sigma_X,
           \quad
          B^{Q1} =  (\sin\omega)\,\sigma_Z + (\cos\omega)\,\sigma_X . \nonumber
\end{align}
This means that any quantum strategy achieving the maximal violation is equivalent to the ideal strategy up to local isometries, and hence if Alice and Bob observe maximal violation of the Bell inequality, they obtain the same statistics as they would from the strategy given in Eq.~\eqref{Eq:optimal quantum strategy}. 

For our analysis, we use essentially the same protocol given in \cite{Wooltorton_2024}. Alice and Bob have access to uncharacterised devices, whose inner workings are unknown to them, but they do not broadcast the outputs maliciously, and they have to operate at spacelike separation. Alice and Bob each also have access to a private random number generator and an authenticated classical channel between them (which however is monitored by Eve). In each round of the protocol, Alice uses her private random number generator to decide whether the round is a test round or a generator round and communicates this information to Bob using an authenticated classical channel after the round. Then Alice and Bob proceed to obtain inputs $x,y \in \{0,1\}$ from their private random number generator and feed them to their local devices, which perform measurements on the shared state and produces outputs $a,b\in \{0,1\}$. 
During test rounds, Alice and Bob independently choose inputs \(X = x \in \{0,1\}\) and \(Y = y \in \{0,1\}\) uniformly at random. In generation rounds, the inputs are fixed deterministically to \(X = 0\) and \(Y = 0\). After a large number of rounds have taken place, they estimate $\langle T_{\theta,\phi,\omega}\rangle$ and if the value is not adequately large, they abort the protocol. If the value is sufficiently large, they proceed to the classical post-processing step, namely error correction and privacy amplification in order to distill the secret key.

Furthermore, when the key generation settings are $x^*=y^*=0$, for $\phi$ sufficiently close to $\frac{\pi}{2}$ and with $\Delta R(\phi) = H\!\left(\frac{1\pm\sin\phi}{2}\right)\in (0,1]$, the key rate of the self-tested strategy is $R \geq 1-\Delta R(\phi)$; the reason is that $A$ is distributed uniformly at random, while $\Pr\{A=B|X=x^*,Y=y^*\} = \frac{1+\sin\phi}{2}$.
Note that $\Delta R(\phi)\rightarrow 0$ as $\phi\rightarrow\frac{\pi}{2}$. At the same time, in the same range of $\phi$ and for $\theta$ and $\omega-\phi$ nonzero but small, $\tau(\theta,\phi,\omega) - \beta(\theta,\phi,\omega) > 0$ but the difference gets arbitrarily small: indeed, when $\theta\rightarrow 0$ or $\omega\rightarrow\phi$, both of Alice's or both of Bob's observables in Eq.~\eqref{Eq:optimal quantum strategy} collapse to the same observable, and it is well-known that in this limit a local behaviour is generated. Hence, in this regime we know that the behaviour is arbitrarily close to being local. 

It is natural to conjecture that the self-test used by \cite{Wooltorton_2024} is robust (which was suggested in \cite[Sec.~4.3]{Le_2023} without working out the proof). Here we show that the device-independent key distillation protocol is robust in the sense that for $\eta := \langle T_{\theta,\phi,\omega} \rangle \geq \tau(\theta,\phi,\omega) - \epsilon$, a key rate $R\geq 1-\Delta R(\phi) - f(\epsilon)$ is guaranteed, where $f(\epsilon)\rightarrow 0$ as $\epsilon\rightarrow 0$. In Section \ref{sec:main}, we state our main result and in Section \ref{sec:methods} we present the proofs. The robust protocol analysis proceeds by first noting that we have two dichotomic projective measurements per party, and this allows us to use Jordan's Lemma to decompose the state and measurements into a block-diagonal system, where each block is a two-qubit system (Subsection \ref{jordan lemma}). Then we use the compactness of the set of quantum strategies to show that any strategy on two qubits that achieves a value of $\eta$ close to $\tau(\theta,\phi,\omega)$ is close to the set of optimal strategies (Subsection \ref{two-qubit-robust:compactness}). 
In Subsection \ref{subsec:robust-key-rate} we show that these elements provide a direct route to robust key rate lower bounds for the protocol in \cite{Wooltorton_2024}, without the need to establish a robust self-test in general. For completeness, and because it may be of independent interest, we give the robust self-test in Subsection \ref{sec:general-robust}. In Section \ref{numericals} we then give concrete numerical bounds on the robust key rate for the protocols from \cite{Wooltorton_2024} and \cite{Farkas_2024}, by bounding the conditional von Neumann entropy using iterated mean divergences and SDP relaxations. We conclude in Section \ref{sec:discussion}.

\section{Main result}
\label{sec:main}
We now come to the main results of this paper. We
give a robust DIQKD protocol with non-zero key rate, based on a robust self-test for the near-attainment of the Tsirelson bound of the Bell correlator $T_{\theta,\phi,\omega}$ when Alice and Bob's systems are known to be qubits.

\begin{theorem}
  \label{thm:robust-key-rate}
  Consider the Bell parameter $T_{\theta,\phi,\omega}$ and assume $\tau(\theta,\phi,\omega) > \beta(\theta,\phi,\omega)$. Then there exists a function $f \equiv f_{\theta,\phi,\omega} : \mathbb{R}_+\rightarrow\mathbb{R}_+$ with $f(0)=0$ and $\lim_{x\searrow 0} f(x)=0$, such that whenever $\eta = \langle T_{\theta,\phi,\omega} \rangle \geq \tau(\theta,\phi,\omega) - \epsilon > \beta(\theta,\phi,\omega)$, the protocol described in Section \ref{sec:prelim} guarantees a device-independent key rate $R \geq 1 - \Delta R(\phi) - f(\epsilon)$. 
  The function $f$ is bounded as $f(x) = O(-x^{1/4}\log x)$ for small enough $x \leq x_0$, where the implicit constant and $x_0$ depend on $\theta$, $\phi$ and $\omega$. 
\end{theorem}


\begin{theorem}
\label{thm:robust-self-test-general}
Consider the Bell parameter $T_{\theta,\phi,\omega}$ and assume $\tau(\theta,\phi,\omega) > \beta(\theta,\phi,\omega)$. There exists a function $g \equiv g_{\theta,\phi,\omega} : \mathbb{R}_+\rightarrow\mathbb{R}_+$ with $g(0)=0$ and $\lim_{x\searrow 0} g(x)=0$, such that for any state $\ket{\Psi}^{\cA\cB\cE}$ and local $\pm 1$-measurements $A^x$ and $B^y$ with Bell violation $\eta = \langle T_{\theta,\phi,\omega} \rangle \geq \tau(\theta,\phi,\omega) - \epsilon > \beta(\theta,\phi,\omega)$, there are local isometries $U$ and $V$ satisfying 
\begin{equation*}
    \left\| \bigl( U\otimes V \otimes \1_{\cE}\bigr)
            \bigl( A^x\otimes B^y \otimes \1_{\cE} \bigr)\ket{\Psi}^{\cA\cB\cE} 
       - \bigl( A^{Qx}\otimes B^{Qy} \otimes \1_{\cE}\bigr) \ket{\psi}^{\cA\cB}
            \otimes \ket{\xi}^{\cE} \right\|_2
    \leq g(\epsilon)
\end{equation*}
for all $x,y\in\{0,1\}$, where $\ket{\psi}^{\cA\cB}$, $A^{Qx}$ and $B^{Qy}$ are given in Eq.~\eqref{Eq:optimal quantum strategy}, and $\ket{\xi}^{\cE} \in \cE$ is a state vector.   The function $g$ is bounded as $g(x) = O(x^{1/4})$ for sufficiently small $x\leq x_0$, where the implicit constant and $x_0$ depend on $\theta$, $\phi$ and $\omega$.  
\end{theorem}


\noindent
We proceed with the proof in the following section. 


\section{Proofs}
\label{sec:methods}
In this section, we prove the main results. 
We first simplify the analysis to break up any general strategy into a direct sum of two-qubit ones in Subsection \ref{jordan lemma}. 
We then leverage the known exact self-test for $T_{\theta,\phi,\omega}$ to argue that for two qubits (more generally in finite dimension) it must be robust, in Subsection \ref{two-qubit-robust:compactness}. 
After that, we give an explicit bound on the robustness of the self-test of the two-qubit strategy in Subsection \ref{two qubit strategy},
based on a sum-of-squares decomposition. 
We then use these insights and Alicki-Fannes-type continuity bounds for the entropy to obtain a robust lower bound on the key rate in Subsection \ref{subsec:robust-key-rate}, proving Theorem \ref{thm:robust-key-rate}.

\subsection{Reduction to qubits}
\label{jordan lemma}
We start with a general measurement strategy for $T_{\theta,\phi,\omega}$, i.e. a joint state $\rho^{\cA\cB} = \tr_{\cE} \psi^{\cA\cB\cE}$ and local observables $A^x = A^x_0-A^x_1$ and $B^y = B^y_0-B^y_1$, for $x,y\in\{0,1\}$. We assume, both in the self-testing setting and the cryptographic setting of secret key generation, that these elements are given to Alice and Bob adversarially by Eve. 
However, both for Alice and Bob's observed statistics and for Eve's post-measurement states, nothing changes if we assume that the $A^x$ and $B^y$ are $\pm 1$-valued (equivalently that the $A^x_a$ and $B^y_b$ are projectors), by way of Naimark's extension theorem (enlarging, if necessary, the Hilbert spaces $\cA$ and $\cB$, respectively).
Jordan's Lemma \cite{Jordan_1875} allows us to simplify this setup significantly: it states that given two Hermitian projection operators in finite-dimensional (or more generally separable) Hilbert space, it is possible to simultaneously block-diagonalise them into blocks of size at most $2$, thereby reducing the Hilbert space to a direct sum of two qubit Hilbert spaces and the two projectors to a direct sum of projectors in these qubit spaces. In fact, by enlarging the Hilbert space, we can without loss of generality assume that all blocks have size $2$. 

We apply this to the operators $A^x$ on $\cA$ and $B^y$ on $\cB$, giving projectors $P_i$ and $Q_j$ of rank two, respectively, satisfying $\1_{\cA}=\sum_i P_i$ and  $\1_{\cB}=\sum_j Q_j$, as well as $[A^x,P_i]=0$ and $[B^y,Q_j]=0$ for all $x$ and $y$. Thus,   
\begin{equation}
A^x_a = \sum_{i} P_i A^x_a,
\qquad
B^y_b = \sum_{j} Q_j B^y_b .
\end{equation}
Let us denote the projected operators as $\tilde{A}^x_{a|i} = P_i A^x_a$ and $\tilde{B}^y_{b|j}= Q_j B^y_b$. By suitable unitaries we can identify $\cA$ with $\cA'\ox\widetilde{\cA}$ and $\cB$ with $\cB'\ox\widetilde{\cB}$, where $\cA'$ is spanned by the orthonormal basis $\{\ket{i}\}$ and $\cB'$ by $\{\ket{j}\}$, and $\widetilde{\cA}$ and $\widetilde{\cB}$ are qubits, such that $P_i = \proj{i}\ox\1_{\widetilde{\cA}}$ and $Q_j = \proj{j}\ox\1_{\widetilde{\cB}}$. Furthermore, 
\begin{equation}
\label{eq:Jordan-observables}
A^x_a = \sum_{i} \proj{i} \ox \tilde{A}^x_{a|i},
\qquad
B^y_b = \sum_{j} \proj{j} \ox \tilde{B}^y_{b|j}.
\end{equation}
In block $(i,j)$, Alice measures $\tilde A^x_{\vert i} := \tilde A^x_{0\vert i} - \tilde A^x_{1\vert i}$
and Bob measures $\tilde B^y_{\vert j} := \tilde B^y_{0\vert j} - \tilde B^y_{1\vert j}$; in particular, Alice's observables depend only on her block index $i$, and Bob's only on his block index $j$. The joint state of Alice, Bob and Eve can also be reduced to a direct sum of two-qubit operators, giving us
\begin{equation}
\label{eq:Jordan-state}
\ket{\Psi} = \sum_{i,j} (P_i\ox Q_j\ox \1_{\cE}) \ket{\Psi}
=: \sum_{i,j} \sqrt{p_{ij}} \ket{i}^{\cA'}\otimes\ket{j}^{\cB'}\otimes\ket{\widetilde{\Psi}_{ij}}^{\widetilde{\cA}\widetilde{\cB}\cE}.
\end{equation}

Suppose that for some $\epsilon>0$ we obtain a Bell violation $\eta = \langle T_{\theta,\phi,\omega} \rangle \geq \tau(\theta,\phi,\omega) -\epsilon > \beta(\theta,\phi,\omega)$. The strategy can be equivalently described by the adversary generating a pair of indices $(i,j)$ with probability $p_{ij}$ and providing the state $\ket{\widetilde{\Psi}_{ij}}$ to Alice and Bob (who get a qubit each), which is subjected to the local measurements with projectors $\tilde{A}^x_{a|i}$ and $\tilde{B}^y_{b|j}$. This assigns to each $2\times 2$-block a Bell score $\eta_{ij}$, and the original Bell score is the average over all the blocks: $\eta = \sum_{i,j} p_{ij} \eta_{ij}$. 
From the Markov inequality we obtain $\Pr\{\eta_{ij} \geq \tau(\theta,\phi,\omega)-\sqrt{\epsilon}\} \geq 1 - \sqrt{\epsilon}$. In other words, as $\epsilon$ gets smaller, the total probability weight of blocks $(i,j)$ where $\eta_{ij} < \tau(\theta,\phi,\omega)-\sqrt{\epsilon}$ is arbitrarily small and so their very contribution to the state $\ket{\Psi}$, and with it the key distillation protocol, is bounded. Thus, we will eventually focus on the blocks $(i,j)$ with $\eta_{ij} \geq \tau(\theta,\phi,\omega)-\sqrt{\epsilon}$, assuming tacitly that the latter is still $> \beta(\theta,\phi,\omega)$, and study how far the corresponding two-qubit strategy is from the optimal strategy.

\subsection{Robustness of the self-test for two-qubit strategies}
\label{two-qubit-robust:compactness}
We can define a distance between two strategies given by tuples $S=(\rho^{AB}, A^x_a, B^y_b)$ and $S'= (\rho^{\prime AB}, A'^x_a, B'^y_b)$, for example by letting $d(S,S') = d_1(\rho^{AB},\rho^{\prime AB}) + \sum_{x,a} d_2(A^x_{a},A^{\prime x}_{a}) + \sum_{y,b} d_2(B^y_{b},B^{\prime y}_{b})$, where $d_1$ is a distance between density matrices (say the trace distance), and $d_2$ is a distance between POVM elements (say the operator norm). With this definition, the set of strategies equipped with 
$d$ forms a metric space $\mathcal{S}$. Further, since we are dealing with two-qubit strategies, the dimension of the Hilbert space is finite, and hence this space of strategies is compact. We now recall/restate the following lemma for compact metric spaces. 

\begin{lemma}
\label{qubit strategy}
  On a compact metric space $X$ with distance function $d$, consider $F:X\rightarrow \mathbb{R}$ a continuous real-valued function, and let $\tau = \sup_{x\in X} F(x)$, which by compactness is attained on the elements of the closed nonempty set $X_{\max} = F^{-1}(\tau) \subset X$. 

  Then there exists a function $\delta:\mathbb{R}_+ \rightarrow \mathbb{R}_+$, which we may assume to be monotonic, with $\delta(0)=0$ and $\lim_{x\searrow 0} \delta(x) = 0$, such that if $F(x)\geq \tau - \epsilon$ for some $\epsilon \geq 0$, it follows that $d(x,X_{\max}) := \inf_{x'\in X_{\max}} d(x,x') \leq \delta(\epsilon)$.
\end{lemma}
\begin{proof}
Assume, by way of contradiction, that there does not exist a function $\delta$ satisfying the conditions stated. This would mean that there exists a $\delta>0$ and points $x_n\in X$ such that $F(x_n) \geq \tau-\frac1n$, while $d(x_n,X_{\max}) \geq \delta$ for all $n$. Then, since $X$ is compact, the sequence $(x_n)$ has a subsequence $(x_{n_k})$ that converges to a point  $x_*\in X$. Since $F$ is a continuous function, $F(x_*) = \lim_{n\rightarrow\infty}F(x_{n_k}) = \tau$ as $k\rightarrow\infty$, and this implies that $x_* \in X_{\max}$. At the same time, the metric distance and hence $d(x,X_{\max})$ are continuous functions, too, so $d(x_*,X_{\max}) \geq \delta > 0$, contradicting the former statement. 
\end{proof}


We apply this to the space of two-qubit strategies and the function $F(S) = \langle T_{\theta,\phi,\omega} \rangle$, which is trilinear in the components of $S = (\rho^{AB}, A^x, B^y)$ and thus continuous. Thanks to \cite{Le_2023, Wooltorton_2024} we already know the set $\mathcal{S}_{\max}$ of maximisers, it is precisely those given by Eq.~\eqref{Eq:optimal quantum strategy}, let us denote it $S_0=(\psi,A^{Qx}_a,B^{Qy}_b)$, up to local (qubit) unitaries $U$ and $V$:
\[
  \mathcal{S}_{\max} = \left\{ S_0^{U,V} = \left((U\ox V)\psi (U\ox V)^\dagger, U A^{Qx}_a U^\dagger, V B^{Qy}_b V^\dagger\right) : U\in \text{SU}(\widetilde{\cA}),\, V\in \text{SU}(\widetilde{\cB}) \right\}.
\]
Hence, from Lemma \ref{qubit strategy} we get the function $\delta$ such that for each block $(i,j)$ with strategy $S_{ij}$ that achieves Bell score $\eta_{ij}\geq\tau(\theta,\phi,\omega)- \epsilon_{ij}$, we get $d(S_{ij},\mathcal{S}_{\max}) \leq \delta(\epsilon_{ij})$. In other words, 
there exist local qubit unitaries $U_{ij}$ and $V_{ij}$ such that $d(S_{ij},S_0^{U_{ij},V_{ij}}) \leq \delta(\epsilon_{ij})$. 

We fix the general statement in the following form.
\begin{lemma}
\label{lemma:two-qubit-robust}
Given $\theta,\phi,\omega$ such that $\tau(\theta,\phi,\omega) > \beta(\theta,\phi,\omega)$, there is a function $g_2:\mathbb{R}_+ \rightarrow \mathbb{R}_+$ with $g_2(0)=0$ and $\lim_{x\searrow 0} g_2(x)=0$, and with the following property:
for any two-qubit strategy $S = (\rho^{\cA\cB},A^x_a,B^y_b)$ with $\rho = \tr_{\cE} \Psi$ and $\eta = \langle T_{\theta,\phi,\omega} \rangle \geq \tau(\theta,\phi,\omega)-\epsilon$, there exist a pure state $\ket{\xi} \in \cE$ and local unitaries $U$ and $V$ satisfying
\[
    \left\| \bigl( U\otimes V \otimes \1_{\cE}\bigr)
            \bigl( A^x\otimes B^y \otimes \1_{\cE}\bigr)\ket{\Psi}^{\cA\cB\cE} 
       - \bigl( A^{Qx}\otimes B^{Qy} \otimes \1_{\cE}\bigr) \ket{\psi}^{\cA\cB}
            \otimes \ket{\xi}^{\cE} \right\|_2
    \leq g_2(\epsilon),
\]
for all $x,y\in\{0,1\}$, where $\ket{\psi}$, $A^{Qx}$ and $B^{Qy}$ are given in Eq.~\eqref{Eq:optimal quantum strategy}. 
\end{lemma}
\begin{proof}
From the above argument, we get that there exist qubit unitaries $U^\dagger$ and $V^\dagger$ such that $d(S,S_0^{U^\dagger,V^\dagger}) \leq \delta(\epsilon)$, in particular $d_1\!\left(\rho,(U\ox V)^\dagger\psi(U\ox V)\right) \leq \delta(\epsilon)$ and for all $x,a,y,b$, $d_2\!\left(A^x_a,U^\dagger A^{Qx}_a U \right),\, d_2\!\left(B^y_b, V^\dagger B^{Qy}_b V \right) \leq \delta(\epsilon)$.

By the equivalence of the trace norm and the purified distance, and invoking Uhlmann's theorem, the former implies that there exists a $\ket{\xi} \in \cE$ with 
\[
  \left\| (U\ox V\ox\1_{\cE})\ket{\Psi}^{\cA\cB\cE} - \ket{\psi}^{\cA\cB}\ox\ket{\xi}^{\cE} \right\|_2 \leq g_2(\epsilon)
\]
for a suitable $g_2$. Indeed, since all observables are norm-bounded by $1$, via the triangle inequality we even get the bound in the lemma for all $x,y$.
\end{proof}

In the next Subsection \ref{two qubit strategy} we shall attempt to give an explicit form of $g_2$, to obtain further insight into its dependence on $\epsilon$.

\subsection{Explicit robust two-qubit self-test via sum-of-squares certificate}
\label{two qubit strategy}


In this section, we construct explicit functions bounding the distance between the near-optimal strategy and an optimal strategy. To start, \cite[Lemma~7]{Wooltorton_2024} states that for a two-qubit state and two dichotomic local qubit observables, for every strategy $(\rho',\tilde{A}'^x,\tilde{B}'^y)$ that achieves $\eta \geq \tau(\theta,\phi,\omega) - \epsilon$, there is another strategy $(\rho'',\tilde{A}''^x,\tilde{B}''^y)$ with reduced parameters of the form
\begin{equation}
  \label{eq:reduced two qubit strategy}
  \rho'' = \sum_{\alpha=0}^{3} \lambda_{\alpha} \proj{\Phi_{\alpha}}, \phantom{====:}
\end{equation}
\begin{equation}
  \label{eq:A-B}
  \begin{aligned}
    \tilde{A}''^{x}  &= (\cos a_x'')\,\sigma_Z + (\sin a_x'')\,\sigma_X, \\
    \tilde{B}''^{y} &= (\cos b_y'')\,\sigma_Z + (\sin b_y'')\,\sigma_X,
  \end{aligned}
\end{equation}
where
\begin{equation}
\label{eq:Bell-states}
\begin{aligned}
  \ket{\Phi_0} &= \tfrac{1}{\sqrt{2}}\bigl(\ket{00} + \ket{11}\bigr), &
  \ket{\Phi_1} &= \tfrac{1}{\sqrt{2}}\bigl(\ket{00} - \ket{11}\bigr), \\
  \ket{\Phi_2} &= \tfrac{1}{\sqrt{2}}\bigl(\ket{01} + \ket{10}\bigr), &
  \ket{\Phi_3} &= \tfrac{1}{\sqrt{2}}\bigl(\ket{01} - \ket{10}\bigr),
\end{aligned}
\end{equation}
 $\lambda_\alpha\ge0$, and $\sum_\alpha \lambda_\alpha=1$.
Although this holds only for Bell inequalities with uniform marginals, we emphasize that this can be assumed without loss of generality as this can be ensured using public one-way communication between Alice and Bob.
 We briefly recap the proof here. Each party can rotate their local basis by a qubit unitary, $O_A$ and $O_B$
respectively, such that
\begin{equation}
\label{eq:A-B-primed}
\begin{aligned}
O_A\tilde{A}'^xO_A^\dagger &= (\cos a'_x) \sigma_Z + (\sin a'_x) \sigma_X, \\
O_B\tilde{B}'^yO_B^\dagger &= (\cos b'_y) \sigma_Z + (\sin b'_y) \sigma_X,
\end{aligned}
\end{equation}
and we write $\hat{\rho}' := (O_A\ox O_B)\rho'(O_A\ox O_B)^\dagger$ for the
correspondingly rotated state. As all these observables anticommute with
$\sigma_Y$, we may symmetrise the state,
\begin{equation}
  \label{eq: rho bar}
  \bar{\rho}'
    := \frac{1}{2}\left(
        \hat{\rho}'
        + (\sigma_Y \otimes \sigma_Y)\, \hat{\rho}' \, (\sigma_Y \otimes \sigma_Y)
    \right),
\end{equation}
with $\tr\!\left[
        \bar{\rho}'\,\bigl(T_{\theta,\phi,\omega}\bigr)      \right]
    =
    \tr\!\left[
        \rho'\,\left(T_{\theta,\phi,\omega} \right)
    \right]
    \ge \tau(\theta,\phi,\omega) -\epsilon.$
Likewise, as the observables are all real, we may keep only the real part of $\bar{\rho}'$ to obtain 
\begin{equation}
\label{Eq: rho bar bar}
    \bar{\bar{\rho}}' 
    := \frac{1}{2}\left(
        \bar{\rho}' + \bar{\rho}'^{*}
    \right).
\end{equation}

In \cite{Pironio_2009} and more recently in \cite{Bhavsar_2023} it was shown that there exist  local unitaries of the form 
\begin{equation}
U_X(\upsilon_X)
  = \exp\left(i\frac{\upsilon_X}{2} \sigma_Y\right)
  = \cos\left(\frac{\upsilon_X}{2}\right)\1
    + i\,\sin\left(\frac{\upsilon_X}{2}\right)\sigma_Y,
\end{equation}
for \( X \in \{A,B\} \), such that for suitable choices of parameters \(\upsilon_X\),
\begin{equation}\label{eq:unitaries for reduced state}
   (U_A \otimes U_B)\,\bar{\bar{\rho}}'\,(U_A^\dagger \otimes U_B^\dagger) = \rho''. 
\end{equation}
We have thus obtained a reduced strategy with measurements in the $X-Z$ plane and the state in the Bell diagonal basis. Further, these unitaries preserve the form of the observables $\tilde{A}'^x$ and $\tilde{B}'^y$, hence we obtain 
\begin{align}
     \tilde{A}''^x&= U_{A}\ O_A\ \tilde{A}'^x \ O_A^{\dagger} \ U_{A}^{\dagger}, \\
     \tilde{B}''^y&=  U_{B}\ O_B\ \tilde{B}'^y \ O_B^{\dagger} \  U_{B}^{\dagger}. 
\end{align}
We now proceed to show that this strategy is robust. 

\begin{lemma}
\label{lemma:robust self-test reduced strategy}
Let $\theta, \phi, \omega \in \mathbb{R}$ satisfy Eq.~\eqref{eq: condition for quantum greater than classical} and assume $\omega$ is restricted to 
\begin{equation}
\label{eq:omega range}
\omega \in \left( \frac{\pi}{2} - \theta \bmod \pi
+ n\pi, \; \frac{\pi}{2} + n\pi \right) \quad \text{for some } n \in \mathbb{Z}.
\end{equation} Then for every reduced strategy $(\rho'',\tilde{A}''^x,\tilde{B}''^y)$ of the
form~\eqref{eq:reduced two qubit strategy} that achieves a Bell score
$\tr(\rho'' T_{\theta,\phi,\omega})\geq\tau(\theta,\phi,\omega)-\epsilon$, for a sufficiently small $\epsilon>0$, there
exist local unitaries $W_A$ and $W_B$ such that
\begin{align}
  W_A\tilde{A}''^0W_A^\dagger &= A^{Q0},\\ 
  \norm{W_A\tilde{A}''^1W_A^\dagger - A^{Q1}} &\le \mathcal{F}_{A_1}(\epsilon), \\
  \norm{W_B\tilde{B}''^yW_B^\dagger - B^{Qy}} &\le \mathcal{F}_{B_y}(\epsilon), \\
  \norm{(W_A\ox W_B)\rho''(W_A\ox W_B)^\dagger - \proj{\Phi_0}}_1
  &\le 2\sum_{\alpha=1}^3\mathcal{F}_\alpha(\epsilon),
\end{align}
where $\mathcal{F}_{A_1}(\epsilon)$, $\mathcal{F}_{B_0}(\epsilon)$, $\mathcal{F}_{B_1}(\epsilon)$, and $\mathcal{F}_\alpha(\epsilon)$ are functions of $\epsilon$ explicitly defined in the proof.
\end{lemma}
\begin{proof}
We begin by introducing the sum-of-squares (SOS) decomposition for the shifted Bell operator \begin{equation}
    \bar{T}_{\theta,\phi,\omega}
    := \tau (\theta, \phi,\omega) \1
    - T_{\theta,\phi,\omega},
\end{equation} 
as originally given in \cite{Barizien2024custombell} and recently used in \cite{Wooltorton_2024}: namely, for a Bell inequality violation in a two-qubit block satisfying $\eta \geq \tau(\theta,\phi,\omega)-\epsilon$, we have
\begin{equation}
\label{eq:SOS-T-bar}
\tr \rho'' \bar{T}_{\theta,\phi,\omega}
= \sum_{\mu=0}^{1} \sum_{\alpha=0}^{3} c_\mu \lambda_{\alpha}
\langle \Phi_{\alpha} | R_{\mu}^{\dagger} R_{\mu} | \Phi_{\alpha} \rangle\le \epsilon, 
\end{equation}
where
\begin{align}
\label{eq: sos R}
    R_0 &= (\sin\theta) \tilde{B}''^0 
    + \cos(\theta + \phi) \tilde{A}''^0 
    - (\cos\phi) \tilde{A}''^1, \\
    R_1 &= (\sin\theta) \tilde{B}''^1 
    + \cos(\theta + \omega) \tilde{A}''^0 
    - (\cos\omega) \tilde{A}''^1,
\end{align}
and
\begin{align}
\label{eq: sos c}
    c_0 = -\frac{\cos\omega \, \cos(\theta + \omega)}{2\sin\theta}, \quad  
    c_1 = \frac{\cos\phi \, \cos(\theta + \phi)}{2\sin\theta},
\end{align}
satisfying $c_0c_1>0$. Also, we denote the SOS operators $R_\mu$ evaluated on the ideal observables of Eq.~\eqref{Eq:optimal quantum strategy}  by $R^Q_\mu$.




As $\lambda_\alpha \ge 0$ and $\Vert{} R_{\mu} \vert{} \Phi_{\alpha} \rangle \Vert{}^2 \ge 0$, every term in the sum-of-squares expansion \eqref{eq:SOS-T-bar} is non-negative.
Hence, we can bound it as follows: 
\[
  c_{\mu}\lambda_{\alpha} \, \| R_{\mu} | \Phi_{\alpha} \rangle \|^2 \leq \epsilon
  \text{ for all } \mu, \alpha.
\]
Further, 
 Eq. \eqref{eq:omega range} along with condition \eqref{eq: condition for quantum greater than classical} ensures that $c_{\mu}>0$. This gives us
\begin{equation} 
\label{sos robust bound}
  \lambda_{\alpha}  \| R_{\mu} | \Phi_{\alpha} \rangle \|^2 \leq\frac{ \epsilon}{c_{\mu} }.
\end{equation}

We first bring the strategy into a simpler form by local unitaries, which
do not change the Bell value since they act on the state and the observables
together.
First, we may assume that the largest eigenvalue of the Bell-diagonal state
belongs to $\ket{\Phi_0}$. If instead $\lambda_\alpha$ with $\alpha\neq0$ is
the largest, we apply the Pauli operator $P_A\in\{\sigma_Z,\sigma_X,\sigma_Y\}$
on Alice's side that maps $\ket{\Phi_\alpha}$ to $\ket{\Phi_0}$ up to a phase;
otherwise we set $P_A:=\1$. This gives the strategy
\begin{equation}
  \rho''_P := (P_A\ox\1)\rho''(P_A\ox\1)^\dagger, \qquad
  \tilde{A}''^x_P := P_A\tilde{A}''^xP_A^\dagger, \qquad
  \tilde{B}''^y_P := \tilde{B}''^y ,
\end{equation}
which is again of the form~\eqref{eq:reduced two qubit strategy}--\eqref{eq:A-B}
and has the same Bell value,  with $\lambda_0\geq\frac14$ and angles $a_x^P$ and $b_y''$.

Second, we rotate Alice's first observable onto $\sigma_Z = A^{Q0}$. Let
$R(\zeta) := \cos\!\left(\frac{\zeta}{2}\right)\1
 - \mathrm{i}\sin\!\left(\frac{\zeta}{2}\right)\sigma_Y$, which maps an
observable $(\cos a)\sigma_Z + (\sin a)\sigma_X$ to
$(\cos(a+\zeta))\sigma_Z + (\sin(a+\zeta))\sigma_X$. We apply
$T := T_A\ox T_B$ with $T_A = T_B := R(\zeta)$ and $\zeta := -a^P_0$, and
denote the rotated strategy by
\begin{equation}
\label{eq:rotated-strategy}
  \rho := T\rho''_PT^\dagger, \qquad
  \tilde{A}^x := T_A\tilde{A}''^x_PT_A^\dagger, \qquad
  \tilde{B}^y := T_B\tilde{B}''^y_PT_B^\dagger .
\end{equation}
Its observables are $\tilde{A}^x = (\cos a_x)\sigma_Z + (\sin a_x)\sigma_X$ and
$\tilde{B}^y = (\cos b_y)\sigma_Z + (\sin b_y)\sigma_X$, with
$a_x := a^P_x + \zeta$ and $b_y := b''_y + \zeta$; in particular $a_0 = 0$, so
 that $\tilde{A}^0 = \sigma_Z = A^{Q0}$.

Applying this lower bound to Eq.~\eqref{sos robust bound}, for $\alpha=0$, we get
\begin{equation}
    \left\| T R_{\mu} T^{\dagger} T\, | \Phi_{0} \rangle \right\| \leq \sqrt{\frac{ 4\epsilon}{c_{\mu} }} := \delta_{\mu}.
\end{equation}



For the choice $\mu = 0$, and for the reduced strategy in Eq. \eqref{eq:reduced two qubit strategy}, we obtain
\begin{equation}
    \left\| T R_{0} T^{\dagger} T \ket{\Phi_{0}} \right\|^2
    =
    \bigl( \sin\theta \, \cos b_0 
    - \cos\phi \, \cos a_1 
    + \cos(\theta + \phi) \bigr)^2
    +
    \bigl( \sin\theta \, \sin b_0 
    - \cos\phi \, \sin a_1 \bigr)^2.
\end{equation}

 Since both the terms are positive, we can obtain the following bounds,

\begin{align}
\label{eq: norm bounds for each term}
     \bigl( \sin\theta \, \cos b_0 
    - \cos\phi \, \cos a_1 
    + \cos(\theta + \phi) \bigr)^2\le & \ \delta_0^2 \\
    \bigl( \sin\theta \, \sin b_0 
    - \cos\phi \, \sin a_1 \bigr)^2 \le &\ \delta_0^2,
\end{align}

Let us define $A := \sin\theta \sin b_0$, and $B := \sin\theta \cos b_0 + \cos(\theta+\phi)$. Then, there exist errors \(e_1,e_2\) satisfying $|e_1|,\;|e_2|\le\delta_0$, such that
\begin{equation}
\label{Eq:relation 1}
    \cos\phi\sin a_1=A+e_1,
\qquad    \cos\phi\cos a_1=B+e_2,
\end{equation}
and 
$e_1^2 + e_2^2 = \norm{TR_0T^\dagger T\ket{\Phi_0}}^2 \leq \delta_0^2$.

Squaring and adding the two identities gives
\begin{align}
\cos^2\phi\sin^2 a_1+\cos^2\phi\cos^2 a_1
&=(A+e_1)^2+(B+e_2)^2 \notag\\
\cos^2\phi
&=(A+e_1)^2+(B+e_2)^2 \notag\\
&=A^2+B^2+2Ae_1+2Be_2+e_1^2+e_2^2 \notag\\
\cos^2\phi-(A^2+B^2)
&=2Ae_1+2Be_2+e_1^2+e_2^2.
\end{align}
Taking absolute values yields
\begin{equation}
\label{eq:sin phi bound 1}
\left|
\cos^2\phi-(A^2+B^2)
\right|
\le
2\delta_0(|A|+|B|)
+\delta_0^2,
\end{equation}
Expanding $A$ and $B$, we get 
\begin{equation}
    A^2+B^2
=
\sin^2\theta
+\cos^2(\theta+\phi)
+2\sin\theta\cos(\theta+\phi)\cos b_0.
\end{equation}

Using
\begin{equation}    
\cos^2\phi
=
\sin^2\theta
+\cos^2(\theta+\phi)
+2\sin\theta\cos(\theta+\phi)\sin\phi,
\end{equation}
we get 
\begin{equation}
 \label{eq:A^2+B^2}
A^2+B^2-\cos^2\phi
=
2\sin\theta\cos(\theta+\phi)
(\cos b_0-\sin\phi).
\end{equation}
Applying Eq. \eqref{eq:A^2+B^2} to inequality \eqref{eq:sin phi bound 1} gives us

\begin{equation}
    \left|
2\sin\theta\cos(\theta+\phi)
(\cos b_0-\sin\phi)
\right|
\le
2\delta_0(|A|+|B|)
+\delta_0^2.
\end{equation}

We recall that the quantum bound is given by $\tau(\theta,\phi,\omega) = |\sin\theta \,\sin(\omega - \phi)\,\sin(\theta + \omega + \phi)|$, whenever 
\begin{equation}
\label{eq:negativity bound}
\cos(\theta + \phi)\,\cos\phi\,\cos(\theta + \omega)\,\cos\omega < 0,
\end{equation}
From these two relations, we can be certain that  $\sin \theta \cos(\theta + \phi)$ is non-zero and hence, we obtain the valid bound, 
\begin{equation}
|\cos b_0 - \sin\phi| 
\le 
\frac{
  2\delta_0 \left( |\sin\theta| ( |\sin b_0| + |\cos b_0| ) + |\cos(\theta + \phi)| \right) + \delta_0^2
}{
  |2\sin\theta \cos(\theta + \phi)|
}.
\end{equation}
%
%
By setting $|\sin b_0| + |\cos b_0|$ to its maximum value $\sqrt{2}$, we get 
\begin{equation} 
\label{eq:cos bo bound}
|\cos b_0 - \sin\phi| \le \frac{
2\delta_0 \left( \sqrt{2}|\sin\theta| \, +| \cos(\theta + \phi)| \right)
+ \delta_0^2}{|2 \sin\theta \,\cos(\theta + \phi)|}=:f^{\cos}_{ b_0}(\delta_0).
\end{equation}

Now we consider $\cos a_1$: recall that
\begin{align}
  \cos\phi\cos a_1
    &= B+e_2\\
    &= \cos \theta \cos \phi + \sin \theta (\cos b_0 - \sin \phi)+e_2\\
  \cos \phi(\cos a_1-\cos \theta)
    &=\sin \theta (\cos b_0 - \sin \phi)+e_2.
\end{align}
Taking absolute values and recalling that $\cos \phi \neq 0$ from Eq.~\eqref{eq:negativity bound}, we get
\begin{equation}
    | \cos a_1-\cos \theta| \leq  \frac{|\sin \theta |f^{\cos}_{ b_0}(\delta_0)+ \delta_0}{|\cos \phi|}=:f^{\cos }_{a_1}(\delta_0).
\end{equation}


Now we obtain a bound on $\sin b_0$. 
Let $e_3 = \cos b_0 - \sin \phi$, where $|e_3| \le f^{\cos}_{ b_0}(\delta_0)$. Expressing $\cos b_0$ in terms of $e_3$:
\begin{equation}
    \cos b_0 = \sin \phi + e_3.
\end{equation}




Squaring both sides, we get:
\begin{align}
1 - \sin^2 b_0 &= \sin^2 \phi + 2e_3 \sin \phi + e_3^2 \\
\cos^2 \phi - \sin^2 b_0 &= 2e_3 \sin \phi + e_3^2 \\
|\cos \phi|^2 - |\sin b_0|^2 &= 2e_3 \sin \phi + e_3^2 \\
\bigl(|\cos \phi| - |\sin b_0|\bigr)\bigl(|\cos \phi| + |\sin b_0|\bigr) &= 2e_3 \sin \phi + e_3^2
\end{align}

Taking absolute values gives:
\begin{equation}
\bigl||\cos \phi| - |\sin b_0|\bigr| \bigl(|\cos \phi| + |\sin b_0|\bigr) \le 2|\sin \phi| f_{b_0}^{\cos}(\delta_0) + \bigl(f_{b_0}^{\cos}(\delta_0)\bigr)^2
\end{equation}

Dividing through by $(|\cos \phi| + |\sin b_0|)$ yields:
\begin{align}
\bigl||\cos \phi| - |\sin b_0|\bigr| &\le \frac{2|\sin \phi| f_{b_0}^{\cos}(\delta_0) + \bigl(f_{b_0}^{\cos}(\delta_0)\bigr)^2}{|\cos \phi| + |\sin b_0|} \\
&\le \frac{2|\sin \phi| f_{b_0}^{\cos}(\delta_0) + \bigl(f_{b_0}^{\cos}(\delta_0)\bigr)^2}{|\cos \phi|}=: f _{b_0}^{\sin } (\delta_0).
\end{align}

or equivalently,
\begin{equation}
\label{eq:sin bo bound}
  \vert{} (-1)^t \cos \phi-\sin b_0 \vert{} \le   f _{b_0}^{\sin } (\delta_0),  
\end{equation}
 for some $t \in \{0, 1\}.$ 


Similarly, let $e_4 = 
\sin b_0 - (-1)^t \cos \phi$, with 
$    |e_4| \le f _{b_0}^{\sin } (\delta_0)$.
Substituting it into Eq.~\eqref{Eq:relation 1} yields
\begin{align}
    \cos \phi \sin a_1 &= \sin \theta \left( (-1)^t \cos \phi + e_4 \right) + e_1 \nonumber \\
     \cos \phi \left( \sin a_1 - (-1)^t \sin \theta \right) &= e_4 \ \sin \theta  + e_1.
\end{align}

Taking absolute values and applying the triangle inequality to the r.h.s. gives
\begin{equation}
    \left| \sin a_1 - (-1)^t \sin \theta \right| \le \frac{|\sin \theta| f _{b_0}^{\sin } (\delta_0) + \delta_0}{|\cos \phi|}=:f^{\sin }_{a_1}(\delta_0).
\end{equation}



Having obtained all the bounds chaining from $R_0$, we can focus on $R_1$. The inequality obtained is 

\begin{equation}
     \left\| T R_{1} T^{\dagger} T \, | \Phi_{0} \rangle \right\|^2
=
\bigl( \sin\theta \, \cos b_1 - \cos\omega \, \cos a_1 + \cos(\theta + \omega) \bigr)^2
+ \bigl( \sin\theta \, \sin b_1 - \cos\omega \, \sin a_1 \bigr)^2\le \delta_1^2,
\end{equation}
which leads us to 
\begin{equation} 
\label{cos b1 bound}
|\cos b_1 - \sin\omega| \le \frac{
2\delta_1 \left( \sqrt{2}|\sin\theta| \, +| \cos(\theta + \omega)| \right)
+ \delta_1^2}{|2 \sin\theta \,\cos(\theta + \omega)|}=:f^{\cos}_{ b_1}(\delta_1)
\end{equation}
and 
\begin{equation}
\vert{}\sin b_1 - (-1)^{t'} \cos \omega\vert{} \le \frac{2|\sin \omega| f_{b_1}^{\cos}(\delta_1) + \bigl(f_{b_1}^{\cos}(\delta_1)\bigr)^2}{|\cos \omega|}=: f^{\sin}_{ b_1}(\delta_1),
\end{equation}
for some $t'\in\{0,1\}$.

Similarly, we obtain another bound for $a_1$ in terms of  $\omega$:
\begin{equation}
      |\cos \theta - \cos a_1| \le \frac{|\sin  \theta| f_{ b_1}^{\cos} (\delta_1)+ \delta_1}{|\cos \omega|}=:g_{ a_1}^{\cos}(\delta_1),
\end{equation}
and 
\begin{equation}
  \left| \sin a_1 - (-1)^{t'} \sin \theta \right| \le \frac{|\sin \theta| f _{b_1}^{\sin } (\delta_1) + \delta_1}{|\cos \omega|}=:g_{ a_1}^{\sin}(\delta_1).
\end{equation} 
The signs $t$ and $t'$ coincide. This is because the bounds defining
$f^{\sin}_{a_1}$ and $g^{\sin}_{a_1}$ give
$\lvert\sin a_1-(-1)^t\sin\theta\rvert\le f^{\sin}_{a_1}(\delta_0)$ and
$\lvert\sin a_1-(-1)^{t'}\sin\theta\rvert\le g^{\sin}_{a_1}(\delta_1)$. If
$t\neq t'$, the triangle inequality would give
\begin{equation}
  2\lvert\sin\theta\rvert
  \le f^{\sin}_{a_1}(\delta_0)+g^{\sin}_{a_1}(\delta_1),
\end{equation}
which gives a contradicts for a sufficiently small $\epsilon$ since $f^{\sin}_{a_1}$ and $g^{\sin}_{a_1}$ are continuous and vanish at $\epsilon = 0$, but $\sin\theta \neq 0$. Hence $t=t'$, and the
same sign appears in all bounds above.

If $t=1$, we apply the local unitary $\sigma_Z\ox\sigma_Z$ to the rotated
strategy in Eq.~\eqref{eq:rotated-strategy}. Hence we write $Z := \sigma_Z$ if $t=1$ and
$Z := \1$ otherwise, giving us \begin{align}
  Z\tilde{A}^0 Z^{\dagger}&= A^{Q0}, \\  
  \| Z\tilde{A}^1Z^{\dagger} - A^{Q1} \| &\leq \min\left\{
    \sqrt{(f^{\sin}_{a_1})^2 + (f^{\cos}_{a_1})^2},
    \sqrt{(g^{\sin}_{a_1})^2 + (g^{\cos}_{a_1})^2} \right\} =: \mathcal{F}_{A_1}(\epsilon), \\
  \| Z\tilde{B}^0Z^{\dagger} - B^{Q0} \| &\leq \sqrt{(f^{\sin}_{ b_0})^2 +( f^{\cos}_{ b_0})^2} =: \mathcal{F}_{B_0}(\epsilon), \\
  \|Z \tilde{B}^1Z^{\dagger} - B^{Q1} \| &\leq \sqrt{(f^{\sin}_{ b_1})^2 + (f^{\cos}_{ b_1})^2} =: \mathcal{F}_{B_1}(\epsilon).
\end{align}

These bounds hold for a sufficiently small $\epsilon$ such that
  \begin{equation}\label{eq:t equivalence}
 2\lvert\sin\theta\rvert>  f^{\sin}_{a_1}(\delta_0)+g^{\sin}_{a_1}(\delta_1) .
\end{equation}
Each of the functions $f^{\cos}_{a_1}$, $f^{\sin}_{a_1}$, $g^{\cos}_{a_1}$,
$g^{\sin}_{a_1}$, $f^{\cos}_{b_y}$ and $f^{\sin}_{b_y}$ is linear in
$\delta_\mu = 2\sqrt{\epsilon/c_\mu}$ up to terms of higher order, and is
therefore $O(\sqrt{\epsilon})$. Consequently,
$\mathcal{F}_{A_1}$, $\mathcal{F}_{B_0}$ and $\mathcal{F}_{B_1}$ are
$O(\sqrt{\epsilon})$ as well.

Next, we derive bounds on the eigenvalues $\lambda_\alpha$ with
$\alpha\in\{1,2,3\}$, using the SOS operators of the new strategy,
$(Z\ox Z)TR_\mu T^\dagger(Z\ox Z)^\dagger$, and its eigenvectors
$(Z\ox Z)T\ket{\Phi_\alpha}$. Since $Z\tilde{A}^0Z^\dagger = A^{Q0}$, the
former differ from the ideal SOS operators only by the observable errors:
\begin{equation}
\label{eq:SOS-flipped}
  (Z\ox Z)TR_0T^\dagger(Z\ox Z)^\dagger
  = R^Q_0 - \sin\theta\,\bigl(B^{Q0} - Z\tilde{B}^0Z^\dagger\bigr)
    + \cos\phi\,\bigl(A^{Q1} - Z\tilde{A}^1Z^\dagger\bigr).
\end{equation}
By unitary invariance, Eq.~\eqref{eq:SOS-flipped}, and the reverse triangle
inequality, together with the bounds on the observables derived above,
\begin{equation}
\begin{aligned}
  \norm{R_0\ket{\Phi_3}}
  &= \norm{(Z\ox Z)TR_0T^\dagger(Z\ox Z)^\dagger\,(Z\ox Z)T\ket{\Phi_3}} \\
  &\geq \norm{R^Q_0\,(Z\ox Z)T\ket{\Phi_3}}
     - \lvert\sin\theta\rvert\,\mathcal{F}_{B_0}(\epsilon)
     - \lvert\cos\phi\rvert\,\mathcal{F}_{A_1}(\epsilon) \\
  &= 2\lvert\sin\theta\rvert
     - \lvert\sin\theta\rvert\,\mathcal{F}_{B_0}(\epsilon)
     - \lvert\cos\phi\rvert\,\mathcal{F}_{A_1}(\epsilon),
\end{aligned}
\end{equation}
where we used $\norm{R^Q_0\,(Z\ox Z)T\ket{\Phi_3}}^2 = 4\sin^2\theta$. For
$\epsilon$ small enough that the right-hand side is positive, combining this
with $c_0\lambda_3\norm{R_0\ket{\Phi_3}}^2\leq\epsilon$ from
Eq.~\eqref{sos robust bound} yields
\begin{equation}
  \lambda_{3}
  \leq \frac{\epsilon}{c_{0}\norm{R_0\ket{\Phi_3}}^2}
  \leq \frac{\epsilon}{c_{0}\bigl(2\lvert\sin\theta\rvert
        - \lvert\sin\theta\rvert\mathcal{F}_{B_0}(\epsilon)
        - \lvert\cos\phi\rvert\mathcal{F}_{A_1}(\epsilon)\bigr)^2} .
\end{equation}

By symmetry, since also
$\norm{R^Q_1\,(Z\ox Z)T\ket{\Phi_3}}^2 = 4\sin^2\theta$, an analogous bound
holds for $R_1$:
\begin{equation}
  \lambda_{3}
  \leq \frac{\epsilon}{c_{1}\norm{R_1\ket{\Phi_3}}^2}
  \leq \frac{\epsilon}{c_{1}\bigl(2\lvert\sin\theta\rvert
        - \lvert\sin\theta\rvert\mathcal{F}_{B_1}(\epsilon)
        - \lvert\cos\omega\rvert\mathcal{F}_{A_1}(\epsilon)\bigr)^2} .
\end{equation}

Using the relations
$\norm{R^Q_0\,(Z\ox Z)T\ket{\Phi_1}}^2 = 4\sin^2\theta\cos^2\bigl(\phi + (-1)^t\zeta\bigr)$
as well as
$\norm{R^Q_1\,(Z\ox Z)T\ket{\Phi_1}}^2 = 4\sin^2\theta\cos^2\bigl(\omega + (-1)^t\zeta\bigr)$,
we obtain in the same way the following bounds on $\lambda_1$:
\begin{align}
  \lambda_{1}
  &\leq \frac{\epsilon}{c_{0}\norm{R_0\ket{\Phi_1}}^2}
   \leq \frac{\epsilon}{c_{0}\bigl(2\lvert\sin\theta\cos\bigl(\phi + (-1)^t\zeta\bigr)\rvert
        - \lvert\sin\theta\rvert\mathcal{F}_{B_0}(\epsilon)
        - \lvert\cos\phi\rvert\mathcal{F}_{A_1}(\epsilon)\bigr)^2}, \\[6pt]
  \lambda_{1}
  &\leq \frac{\epsilon}{c_{1}\norm{R_1\ket{\Phi_1}}^2}
   \leq \frac{\epsilon}{c_{1}\bigl(2\lvert\sin\theta\cos\bigl(\omega + (-1)^t\zeta\bigr)\rvert
        - \lvert\sin\theta\rvert\mathcal{F}_{B_1}(\epsilon)
        - \lvert\cos\omega\rvert\mathcal{F}_{A_1}(\epsilon)\bigr)^2}.
\end{align}
Similarly, using the relations
$\norm{R^Q_0\,(Z\ox Z)T\ket{\Phi_2}}^2 = 4\sin^2\theta\sin^2\bigl(\phi + (-1)^t\zeta\bigr)$
and likewise 
$\norm{R^Q_1\,(Z\ox Z)T\ket{\Phi_2}}^2 = 4\sin^2\theta\sin^2\bigl(\omega + (-1)^t\zeta\bigr)$,
we obtain
\begin{align}
  \lambda_{2}
  &\leq \frac{\epsilon}{c_{0}\norm{R_0\ket{\Phi_2}}^2}
   \leq \frac{\epsilon}{c_{0}\bigl(2\lvert\sin\theta\sin\bigl(\phi + (-1)^t\zeta\bigr)\rvert
        - \lvert\sin\theta\rvert\mathcal{F}_{B_0}(\epsilon)
        - \lvert\cos\phi\rvert\mathcal{F}_{A_1}(\epsilon)\bigr)^2}, \\[6pt]
  \lambda_{2}
  &\leq \frac{\epsilon}{c_{1}\norm{R_1\ket{\Phi_2}}^2}
   \leq \frac{\epsilon}{c_{1}\bigl(2\lvert\sin\theta\sin\bigl(\omega + (-1)^t\zeta\bigr)\rvert
        - \lvert\sin\theta\rvert\mathcal{F}_{B_1}(\epsilon)
        - \lvert\cos\omega\rvert\mathcal{F}_{A_1}(\epsilon)\bigr)^2}.
\end{align}

The bounds on $\lambda_1$ and $\lambda_2$ depend on the device through the angle $(-1)^t\zeta$. In order to obtain device-independent
bounds depending only on $\theta,\phi,\omega$ and $\epsilon$, we lower bound
$\max\bigl(\lvert\cos(\phi + (-1)^t\zeta)\rvert,\lvert\cos(\omega + (-1)^t\zeta)\rvert\bigr)$,
which appears in the bounds on $\lambda_1$, and the same maximum with sines,
which appears in the bounds on $\lambda_2$. Let us define $a := \phi+(-1)^t\zeta$ and $b := \omega+(-1)^t\zeta$, we have,
\begin{equation}
 \cos^2a+\cos^2b = 1+\cos(a+b)\cos(\omega-\phi)
 \geq 1-\lvert\cos(\omega-\phi)\rvert   ,
\end{equation}
 where the inequality holds since $\lvert\cos(a+b)\rvert\leq1$. This gives \[\max(\lvert\cos a\rvert,\lvert\cos b\rvert)\geq\sqrt{\tfrac12\bigl(1-\lvert\cos(\omega-\phi)\rvert\bigr)}=: m,\]

and $m$ is positive since $\sin(\omega-\phi)\neq0$ by
Eq.~\eqref{eq: condition for quantum greater than classical}. Similarly, we obtain $\sin^2a+\sin^2b = 1-\cos(a+b)\cos(\omega-\phi)$, which gives
\[\max(\lvert\sin a\rvert,\lvert\sin b\rvert)\geq m.\] 

Hence, for every device,  at least one of the two bounds on
$\lambda_1$ above has an ideal value of at least
$2\lvert\sin\theta\rvert\,m$: the bound from $R_0$ if
$\lvert\cos a\rvert\geq m$, and the bound from $R_1$ if
$\lvert\cos b\rvert\geq m$. To turn this into a bound that does not depend on
the device, let, with $\chi_0 := \phi$ and $\chi_1 := \omega$,
\begin{equation}
\label{eq:D-mu}
  D_\mu(\epsilon) := 2\lvert\sin\theta\rvert\,m
     - \lvert\sin\theta\rvert\,\mathcal{F}_{B_\mu}(\epsilon)
     - \lvert\cos\chi_\mu\rvert\,\mathcal{F}_{A_1}(\epsilon),
  \qquad \mu\in\{0,1\},
\end{equation}
and restrict to $\epsilon$ small enough that
\begin{equation}
\label{eq:conditon for positivity}
  D_0(\epsilon)>0 \qquad\text{and}\qquad D_1(\epsilon)>0 .
\end{equation}
Suppose first that $\lvert\cos a\rvert\geq m$. Then the bracket in the denominator of the
bound on $\lambda_1$ from $R_0$ is at least $D_0(\epsilon)>0$, so its square
is at least $D_0(\epsilon)^2$, and that bound gives
$\lambda_1\leq\epsilon/\bigl(c_0D_0(\epsilon)^2\bigr)$. If instead
$\lvert\cos b\rvert\geq m$, the bound from $R_1$ gives in the same way
$\lambda_1\leq\epsilon/\bigl(c_1D_1(\epsilon)^2\bigr)$. Since one of the two
cases occurs for every device, but it is not known which, we obtain
\begin{equation}
\label{eq:F-alpha}
  \lambda_1 \leq \mathcal{F}_1(\epsilon)
  := \max\left(\frac{\epsilon}{c_0D_0(\epsilon)^2},\,
               \frac{\epsilon}{c_1D_1(\epsilon)^2}\right).
\end{equation}
The same argument, with $\max(\lvert\sin a\rvert,\lvert\sin b\rvert)\geq m$
in place of $\max(\lvert\cos a\rvert,\lvert\cos b\rvert)\geq m$, gives
$\lambda_2\leq\mathcal{F}_2(\epsilon) := \mathcal{F}_1(\epsilon)$.

The two bounds on $\lambda_3$ above contain no device-dependent quantities,
so both hold for every device whenever 
$2\lvert\sin\theta\rvert - \lvert\sin\theta\rvert\,\mathcal{F}_{B_\mu}(\epsilon)
 - \lvert\cos\chi_\mu\rvert\,\mathcal{F}_{A_1}(\epsilon)>0$, for
$\mu\in\{0,1\}$.  This follows from $D_0(\epsilon)>0$ and $D_1(\epsilon)>0$ since $m \le \frac{1}{\sqrt2} < 1$. Since both bounds hold, we may use
the smaller one:
\begin{equation}
\label{eq:F-3}
  \lambda_3 \leq \mathcal{F}_3(\epsilon)
  := \min_{\mu\in\{0,1\}}
     \frac{\epsilon}{c_\mu\bigl(2\lvert\sin\theta\rvert
       - \lvert\sin\theta\rvert\,\mathcal{F}_{B_\mu}(\epsilon)
       - \lvert\cos\chi_\mu\rvert\,\mathcal{F}_{A_1}(\epsilon)\bigr)^2}.
\end{equation}
The functions $\mathcal{F}_1$, $\mathcal{F}_2$ and $\mathcal{F}_3$ depend
only on $\theta,\phi,\omega$ and $\epsilon$. Moreover,
$\mathcal{F}_\alpha(\epsilon) = O(\epsilon)$, since
$D_\mu(\epsilon)\to2\lvert\sin\theta\rvert\,m>0$ as $\epsilon\to0$.
Since the eigenvalues sum to one, this yields
\begin{equation}
  \lambda_0 \ge 1 - \sum_{\alpha=1}^{3} \mathcal{F}_{\alpha}(\epsilon),
\end{equation}
which is non-trivial for $\epsilon$ small enough that
$\sum_{\alpha=1}^{3}\mathcal{F}_{\alpha}(\epsilon) < 1$.

Finally, we collect the local unitaries used above. Let $P_A$ be the Pauli
operator chosen to make $\lambda_0$ the largest eigenvalue (with $P_A := \1$
if no such choice was needed), $T_A = T_B = R(\zeta)$ the rotation bringing
Alice's first observable onto $\sigma_Z$, and $Z := \sigma_Z$ if $t=1$ and
$Z := \1$ otherwise. 
Thus, the unitaries
\begin{equation}
  W_A := Z\,T_A\,P_A, \qquad W_B := Z\,T_B, \qquad W := W_A\ox W_B,
\end{equation}
map the reduced strategy $(\rho'',\tilde{A}''^x,\tilde{B}''^y)$ of the lemma to the new strategy
$\bigl((Z\ox Z)\rho(Z\ox Z)^\dagger,\,Z\tilde{A}^xZ^\dagger,\,Z\tilde{B}^yZ^\dagger\bigr)$:

\begin{equation}\begin{split}
  (Z\ox Z)\,\rho\,(Z\ox Z)^\dagger 
     &= (W_A\ox W_B)\,\rho''\,(W_A\ox W_B)^\dagger, \\
  Z \tilde{A}^x Z^{\dagger} &= W_A\tilde{A}''^xW_A^\dagger, \qquad
  Z\tilde{B}^yZ^{\dagger} = W_B\tilde{B}''^yW_B^\dagger .
\end{split}\end{equation} 
By construction, the observables $W_A\tilde{A}''^xW_A^\dagger$ and $W_B\tilde{B}''^yW_B^\dagger$ are those for which the bounds on the angles were derived, so they satisfy the bounds on the observables stated in the lemma.

For the state, note that 
$W\rho''W^\dagger = (Z\ox Z)\,T\rho''_PT^\dagger\,(Z\ox Z)^\dagger
 = \sum_\alpha\lambda_\alpha\proj{\Phi'_\alpha}$
with the orthonormal vectors $\ket{\Phi'_\alpha} := (Z\ox Z)T\ket{\Phi_\alpha}$.
Since $T$ and $Z\ox Z$ leave $\ket{\Phi_0}$ invariant, $\ket{\Phi'_0} =
\ket{\Phi_0}$, and therefore
\begin{equation}
  W\rho''W^\dagger - \proj{\Phi_0}
  = -(1-\lambda_0)\proj{\Phi_0}
    + \sum_{\alpha=1}^{3}\lambda_\alpha\proj{\Phi'_\alpha}
\end{equation}
is diagonal in an orthonormal basis. Its trace norm is the sum of the absolute
values of its eigenvalues, so that
\begin{equation}
  \norm{W\rho''W^\dagger - \proj{\Phi_0}}_1
  = (1-\lambda_0) + \sum_{\alpha=1}^{3}\lambda_\alpha
  = 2(1-\lambda_0)
  \leq 2\sum_{\alpha=1}^{3}\mathcal{F}_{\alpha}(\epsilon).
\end{equation}

This completes the proof. 

\end{proof}
We have shown a robust self-test for the strategies given by Eq.~\eqref{eq:reduced two qubit strategy}. We now extend this robust self-test to all two-qubit strategies.

\begin{theorem}
\label{lemma: two qubit general strategy}
For all $\theta,\phi,\omega$ satisfying
Eqs.~\eqref{eq: condition for quantum greater than classical}
and~\eqref{eq:omega range}, there exists $\epsilon_0>0$ such that for all
 $\epsilon\in[0,\epsilon_0]$, inequalities  \eqref{eq:conditon for positivity}  and \eqref{eq:t equivalence} hold, and  $\sum_{\alpha=1}^3\mathcal{F}_\alpha(\epsilon)<\frac14$, and
for every two-qubit strategy $(\rho',\tilde{A}'^x,\tilde{B}'^y)$, there exist qubit unitaries $U$ and $V$ such that,
    if
    $\tr\bigl(\rho'T_{\theta,\phi,\omega}\bigr)\geq\tau(\theta,\phi,\omega)-\epsilon$, then,
\begin{align}
  U\tilde{A}'^0U^\dagger &= A^{Q0}, \\
  \norm{U\tilde{A}'^1U^\dagger - A^{Q1}} &\le \mathcal{F}_{A_1}(\epsilon), \\
  \norm{V\tilde{B}'^0V^\dagger - B^{Q0}} &\le \mathcal{F}_{B_0}(\epsilon), \\
  \norm{V\tilde{B}'^1V^\dagger - B^{Q1}} &\le \mathcal{F}_{B_1}(\epsilon), \\
  \norm{(U\ox V)\rho'(U\ox V)^\dagger - \proj{\Phi_0}}_1
  &\le 4\sqrt{\sum_{\alpha=1}^3\mathcal{F}_\alpha(\epsilon)}.
\end{align}
\end{theorem}

\begin{proof}

The functions $\mathcal{F}_{A_1}$, $\mathcal{F}_{B_0}$, $\mathcal{F}_{B_1}$
and $\mathcal{F}_\alpha$, as well as $f^{\sin}_{a_1}(\delta_0)$ and
$g^{\sin}_{a_1}(\delta_1)$, are continuous in $\epsilon$ and vanish at
$\epsilon=0$. Furthermore, $D_0(\epsilon)$ and $D_1(\epsilon)$ tend to the
positive constant $2\lvert\sin\theta\rvert\,m$ as $\epsilon\to0$, and the
left-hand side $2\lvert\sin\theta\rvert$ of~\eqref{eq:t equivalence} is a
positive constant. Hence, by continuity, there is $\epsilon_0>0$ such that for
all $\epsilon\in[0,\epsilon_0]$ the
inequalities~\eqref{eq:conditon for positivity} and~\eqref{eq:t equivalence}
hold and $\sum_{\alpha=1}^3\mathcal{F}_\alpha(\epsilon)<\frac14$.

Now we will obtain the bounds on the two-qubit strategy.
We can reduce the general strategy $(\rho',\tilde{A}'^x,\tilde{B}'^y)$ to a strategy of the form given in Eqs. \eqref{eq:reduced two qubit strategy} and \eqref{eq:A-B}, as shown in Subsection \ref{two qubit strategy}. This gives us a reduced two qubit strategy $(\rho'',\tilde{A}''^x,\tilde{B}''^y)$ with $\tr(\rho'' T_{\theta,\phi,\omega})\geq\tau(\theta,\phi,\omega)-\epsilon$ where $T_{\theta,\phi,\omega}$ is evaluated on the observables of  $(\rho'',\tilde{A}''^x,\tilde{B}''^y)$.

Since we have chosen an $\epsilon$ that satisfies the conditions of Lemma~\ref{lemma:robust self-test reduced strategy}, it gives us local
unitaries $W_A$ and $W_B$ such that
\begin{align}
  W_A\tilde{A}''^0W_A^\dagger &= A^{Q0},\\ 
  \norm{W_A\tilde{A}''^1W_A^\dagger - A^{Q1}} &\le \mathcal{F}_{A_1}(\epsilon),
  \label{eq:gen-reduced-A}\\
  \norm{W_B\tilde{B}''^yW_B^\dagger - B^{Qy}} &\le \mathcal{F}_{B_y}(\epsilon),   \label{eq:gen-reduced-Bob}\\
  \norm{(W_A\ox W_B)\rho''(W_A\ox W_B)^\dagger - \proj{\Phi_0}}_1
  &\le 2\sum_{\alpha=1}^3\mathcal{F}_\alpha(\epsilon).
  \label{eq:gen-reduced-B}
\end{align}

Define
\begin{equation}
  U := W_AU_AO_A, \qquad V := W_BU_BO_B,
\end{equation}
 where $O_A$, $O_B$ are the qubit unitaries rotating the observables into the
$X$--$Z$ plane, Eq.~\eqref{eq:A-B-primed}, and $U_A$, $U_B$ are the rotations
about the $Y$-axis that bring the averaged state $\bar{\bar{\rho}}'$ into
Bell-diagonal form, Eq.\eqref{eq:unitaries for reduced state}, as in Subsection~\ref{two qubit strategy}. By construction, $U\tilde{A}'^xU^\dagger = W_A\tilde{A}''^xW_A^\dagger$ and
$V\tilde{B}'^yV^\dagger = W_B\tilde{B}''^yW_B^\dagger$, so~\eqref{eq:gen-reduced-A}
and~\eqref{eq:gen-reduced-Bob} give
\begin{align}
  U\tilde{A}'^0U^\dagger &= A^{Q0}, &
  \norm{U\tilde{A}'^1U^\dagger - A^{Q1}} &\le \mathcal{F}_{A_1}(\epsilon), &
  \norm{V\tilde{B}'^yV^\dagger - B^{Qy}} &\le \mathcal{F}_{B_y}(\epsilon).
\end{align}

Having obtained the bound on the observables, we focus on the state.

For a state $\sigma$ and a pure state $\ket{\psi}$, we write
$F(\sigma,\proj{\psi}) := \bra{\psi}\sigma\ket{\psi}$ for their (squared) fidelity, and
recall that
\begin{equation}
\label{eq:fid-trace}
  2\bigl(1 - F(\sigma,\proj{\psi})\bigr) \le \norm{\sigma - \proj{\psi}}_1
  \le 2\sqrt{1 - F(\sigma,\proj{\psi})},
\end{equation}
where the first
follows from
$\frac12\norm{\sigma - \proj{\psi}}_1
 \geq \tr\bigl[\proj{\psi}\bigl(\proj{\psi}-\sigma\bigr)\bigr]
 = 1 - \bra{\psi}\sigma\ket{\psi}$,
since $\ket{\psi}$ is pure. The second inequality is the Fuchs--van de Graaf inequality for a pure
state.

Let $K := W_AU_A\ox W_BU_B$, so that $U\ox V = K(O_A\ox O_B)$, and let
\begin{equation}
  \ket{\hat{\Phi}_0} := K^\dagger\ket{\Phi_0}.
\end{equation}
We have $\norm{\bar{\bar{\rho}}' - \proj{\hat{\Phi}_0}}_1
 = \norm{(W_A\ox W_B)\rho''(W_A\ox W_B)^\dagger - \proj{\Phi_0}}_1$ 
by unitary invariance of the trace norm, so by Eq.~\eqref{eq:gen-reduced-B},
\begin{subequations}
\label{eq:chain-gen}
\begin{align}
  2\sum_{\alpha=1}^3\mathcal{F}_\alpha(\epsilon)
  &\ge \norm{\bar{\bar{\rho}}' - \proj{\hat{\Phi}_0}}_1
   = \norm{\frac{\bar{\rho}' + \bar{\rho}'^{*}}{2} - \proj{\hat{\Phi}_0}}_1 ,
  \label{eq:chain-gen-a}\\
  \sum_{\alpha=1}^3\mathcal{F}_\alpha(\epsilon)
  &\ge 1 - F\!\left(\frac{\bar{\rho}' + \bar{\rho}'^{*}}{2},
                   \proj{\hat{\Phi}_0}\right)
  \label{eq:chain-gen-b}\\
  &= \tfrac12\bigl(1 - F(\bar{\rho}',\proj{\hat{\Phi}_0})\bigr)
   + \tfrac12\bigl(1 - F(\bar{\rho}'^{*},\proj{\hat{\Phi}_0})\bigr)
  \label{eq:chain-gen-c}\\
  &\ge \tfrac12\bigl(1 - F(\bar{\rho}',\proj{\hat{\Phi}_0})\bigr)
  \label{eq:chain-gen-d}\\
  &\ge \tfrac14\bigl(1 - F(\hat{\rho}',\proj{\hat{\Phi}_0})\bigr).
  \label{eq:chain-gen-e}
\end{align}
\end{subequations}

The equality in Eq.~\eqref{eq:chain-gen-a} holds by the definition of $\bar{\bar{\rho}}'$ in Eq.~\eqref{Eq: rho bar bar}, and Eq.~\eqref{eq:chain-gen-b} follows from the first inequality in Eq.~\eqref{eq:fid-trace}. The equality in Eq.~\eqref{eq:chain-gen-c} holds by the linearity of fidelity for pure states, and Eq.~\eqref{eq:chain-gen-d} holds because $F(\bar{\rho}'^{*},\proj{\hat{\Phi}_0})\le1$. Finally, Eq.~\eqref{eq:chain-gen-e} follows in the same way from the definition of $\bar{\rho}'$ in Eq.~\eqref{eq: rho bar}.

Finally, since $\hat{\rho}' = (O_A\ox O_B)\rho'(O_A\ox O_B)^\dagger$ and
$K(O_A\ox O_B) = U\ox V$,
\begin{equation}
  F(\hat{\rho}',\proj{\hat{\Phi}_0})
  = \bra{\Phi_0}K\hat{\rho}'K^\dagger\ket{\Phi_0}
  = \bra{\Phi_0}(U\ox V)\rho'(U\ox V)^\dagger\ket{\Phi_0}.
\end{equation}

Combining~\eqref{eq:chain-gen} with the second inequality
in~\eqref{eq:fid-trace}, we obtain
\begin{equation}
\begin{aligned}
  \norm{(U\ox V)\rho'(U\ox V)^\dagger - \proj{\Phi_0}}_1
  &\le 2\sqrt{1 - F\bigl((U\ox V)\rho'(U\ox V)^\dagger,\proj{\Phi_0}\bigr)} \\
  &\le 2\sqrt{4\sum_{\alpha=1}^3\mathcal{F}_\alpha(\epsilon)}
   = 4\sqrt{\sum_{\alpha=1}^3\mathcal{F}_\alpha(\epsilon)} .
\end{aligned}
\end{equation}
Since the same unitaries $U$ and $V$ appear in the bounds on the observables,
this completes the proof.
\end{proof}

Combining Lemma \ref{lemma:robust self-test reduced strategy} and Theorem \ref{lemma: two qubit general strategy}, we get a two-qubit robust self-test bound.

\begin{theorem}
\label{thm:two-qubit}
Let $\tilde{A}^x$, $\tilde{B}^y$ be local qubit observables acting on a tripartite state $\ket{\tilde{\Psi}}^{\cA\cB\cE} $. Let $\rho$ be the reduced state of Alice and Bob  such that $\tr(\rho T_{\theta,\phi,\omega})\ge \tau(\theta,\phi,\omega)- \epsilon$. Then there exist local qubit unitaries  $U$ and $V$ and a unit vector $\vert\xi\rangle^{\cE}$ such that the following robust self-test bound holds for a valid range of $\epsilon$ and the parameters $(\theta, \phi, \omega)$, as required in Theorem \ref{lemma: two qubit general strategy}:
\begin{align}
\label{eq:two_qubit_bound}
\bigl\|(U\otimes V\otimes \1_{\cE})
\bigl(\tilde{A}^{x}\otimes &\tilde{B}^{y} \otimes \1_{\cE}\bigr)
\ket{\tilde{\Psi}}^{\cA\cB\cE}- \bigl( A^{Qx}\otimes B^{Qy} \otimes \1_{\cE}\bigr)
\ket{\Phi_0}^{\cA\cB}\otimes \ket{\xi}^{\cE}
\bigr\|_2  \notag\\
&\le
2\sqrt{2\sum_{\alpha=1}^3\mathcal{F}_{\alpha}(\epsilon)}
+ \mathcal{F}_{A_x}(\epsilon)+\mathcal{F}_{B_y}(\epsilon)
=: g_2(\epsilon).
\end{align}
\end{theorem}

\begin{proof}
Let $U$ and $V$ be the qubit unitaries provided by
Theorem~\ref{lemma: two qubit general strategy} for the strategy
$(\rho,\tilde{A}^x,\tilde{B}^y)$. Since $U$ and $V$ are unitary,
\begin{equation}
\label{eq:move-rotation}
  (U\ox V\ox\1_{\cE})(\tilde{A}^{x}\ox\tilde{B}^{y}\ox\1_{\cE})
  = \bigl(U\tilde{A}^{x}U^\dagger\ox V\tilde{B}^{y}V^\dagger\ox\1_{\cE}\bigr)
    (U\ox V\ox\1_{\cE}).
\end{equation}
We add and subtract the convenient term
$(A^{Qx}\ox B^{Qy}\ox\1_{\cE})(U\ox V\ox\1_{\cE})\ket{\tilde{\Psi}}^{\cA\cB\cE}$.
Applying the triangle inequality, we get
\begin{align}
\bigl\|(U\otimes V\otimes \1_{\cE})
 &\bigl(\tilde{A}^{x}\otimes \tilde{B}^{y} \otimes \1_{\cE}\bigr)
 \ket{\tilde{\Psi}}^{\cA\cB\cE}
 -\bigl( A^{Qx}\otimes B^{Qy} \otimes \1_{\cE}\bigr)
 \ket{\Phi_0}^{\cA\cB}\otimes \ket{\xi}^{\cE}\bigr\|_2 \notag \\
 &\leq \bigl\|
 \bigl(U\tilde{A}^{x}U^\dagger\ox V\tilde{B}^{y}V^\dagger\ox\1_{\cE}
 - A^{Qx}\ox B^{Qy}\ox\1_{\cE}\bigr)
 (U\ox V\ox\1_{\cE})\ket{\tilde{\Psi}}^{\cA\cB\cE}
 \bigr\|_2 \notag \\
 &\phantom{=}+ \bigl\|
 (A^{Qx}\ox B^{Qy}\ox\1_{\cE})
 \bigl((U\ox V\ox\1_{\cE})\ket{\tilde{\Psi}}^{\cA\cB\cE}
 - \ket{\Phi_0}^{\cA\cB}\ox\ket{\xi}^{\cE}\bigr)
 \bigr\|_2 .
\end{align}
We bound both terms separately. For the first term, we obtain
\begin{subequations}\label{eq:main}
\begin{align}
\text{1st term}
&\leq \bigl\|
  U\tilde{A}^{x}U^\dagger\ox V\tilde{B}^{y}V^\dagger\ox\1_{\cE}
  - A^{Qx}\ox B^{Qy}\ox\1_{\cE}\bigr\|
  \cdot \bigl\|(U\ox V\ox\1_{\cE})\ket{\tilde{\Psi}}^{\cA\cB\cE}\bigr\|_2
  \label{eq:main-a}\\
&= \bigl\|
  \bigl(U\tilde{A}^{x}U^\dagger\ox V\tilde{B}^{y}V^\dagger
  - A^{Qx}\ox B^{Qy}\bigr)\ox\1_{\cE}\bigr\|
  \label{eq:main-b}\\
&= \bigl\|
  \bigl(U\tilde{A}^{x}U^\dagger - A^{Qx}\bigr)\ox V\tilde{B}^{y}V^\dagger
  + A^{Qx}\ox\bigl(V\tilde{B}^{y}V^\dagger - B^{Qy}\bigr)\bigr\|
  \label{eq:main-c}\\
&\leq \bigl\|U\tilde{A}^{x}U^\dagger - A^{Qx}\bigr\|
  + \bigl\|V\tilde{B}^{y}V^\dagger - B^{Qy}\bigr\|
  \label{eq:main-d}\\
&\leq \mathcal{F}_{A_x}(\epsilon) + \mathcal{F}_{B_y}(\epsilon).
  \label{eq:main-e}
\end{align}
\end{subequations}
Here, \eqref{eq:main-a} uses submultiplicativity; \eqref{eq:main-b} uses that
$(U\ox V\ox\1_{\cE})\ket{\tilde{\Psi}}^{\cA\cB\cE}$ is a unit vector and
factors out $\1_{\cE}$; \eqref{eq:main-c} uses
$\norm{X\ox\1_{\cE}} = \norm{X}$ and adds and subtracts
$A^{Qx}\ox V\tilde{B}^{y}V^\dagger$; \eqref{eq:main-d} follows from the
triangle inequality and $\norm{M\ox N} = \norm{M}\,\norm{N}$, since
$V\tilde{B}^{y}V^\dagger$ and $A^{Qx}$ have norm one; and \eqref{eq:main-e}
follows from Theorem~\ref{lemma: two qubit general strategy}.

For the second term, since $A^{Qx}\ox B^{Qy}\ox\1_{\cE}$ is unitary,
\begin{equation}
\label{eq:second-a}
  \text{2nd term}
  = \bigl\|(U\ox V\ox\1_{\cE})\ket{\tilde{\Psi}}^{\cA\cB\cE}
    - \ket{\Phi_0}^{\cA\cB}\ox\ket{\xi}^{\cE}\bigr\|_2 .
\end{equation}

By~\eqref{eq:chain-gen} and the identity
$F(\hat{\rho},\proj{\hat{\Phi}_0})
 = \bra{\Phi_0}(U\ox V)\rho(U\ox V)^\dagger\ket{\Phi_0}$
from the proof of Theorem~\ref{lemma: two qubit general strategy}, we have
\begin{equation}
\label{eq:F-rotated}
  F := \bra{\Phi_0}(U\ox V)\rho(U\ox V)^\dagger\ket{\Phi_0}
  \geq 1 - 4\sum_{\alpha=1}^3\mathcal{F}_{\alpha}(\epsilon).
\end{equation}
Since $(U\ox V\ox\1_{\cE})\ket{\tilde{\Psi}}^{\cA\cB\cE}$
is a purification of $(U\ox V)\rho(U\ox V)^\dagger$, Uhlmann's theorem gives a
unit vector $\ket{\xi}^{\cE}$ such that
\begin{equation}
  \left|\left\langle \Phi_0^{\cA\cB}\otimes\xi^{\cE}
  \,\middle|\, (U\ox V\ox\1_{\cE})\tilde{\Psi}^{\cA\cB\cE}\right\rangle\right|
  = \sqrt{F}.
\end{equation}
Choosing the phase of $\ket{\xi}^{\cE}$ so that this inner product is real and
nonnegative, and using $\sqrt{F}\geq F$ for $F\in[0,1]$, we obtain
\begin{equation}
\begin{aligned}
  \bigl\|(U\ox V\ox\1_{\cE})\ket{\tilde{\Psi}}^{\cA\cB\cE}
    - \ket{\Phi_0}^{\cA\cB}\ox\ket{\xi}^{\cE}\bigr\|_2^2
  &= 2 - 2\operatorname{Re}\left\langle \Phi_0^{\cA\cB}\otimes\xi^{\cE}
     \,\middle|\, (U\ox V\ox\1_{\cE})\tilde{\Psi}^{\cA\cB\cE}\right\rangle\\ 
  &= 2 - 2\sqrt{F} \leq 2(1-F)
  \leq 8\sum_{\alpha=1}^3\mathcal{F}_{\alpha}(\epsilon),
\end{aligned}
\end{equation}
by~\eqref{eq:F-rotated}. Therefore,
\begin{equation}
  \text{2nd term} \leq 2\sqrt{2\sum_{\alpha=1}^3\mathcal{F}_{\alpha}(\epsilon)}.
\end{equation}
Combining both bounds, we obtain
\begin{equation}
\bigl\|(U\otimes V\otimes \1_{\cE})
 \bigl(\tilde{A}^{x}\otimes \tilde{B}^{y} \otimes \1_{\cE}\bigr)
 \ket{\tilde{\Psi}}^{\cA\cB\cE}
 -\bigl( A^{Qx}\otimes B^{Qy} \otimes \1_{\cE}\bigr)
 \ket{\Phi_0}^{\cA\cB}\otimes \ket{\xi}^{\cE}\bigr\|_2
 \leq g_2(\epsilon).
\end{equation} 
This completes the proof. Further, each of the terms in $g_2(\epsilon)$ scales as $O(\epsilon^{1/2})$, hence $g_2(\epsilon)$ also scales as  $O(\epsilon^{1/2})$. 
\end{proof}

\subsection{Non-zero key rate from near-maximum Bell violation}
\label{subsec:robust-key-rate}
We start by analysing the simpler case when Alice and Bob have qubit systems. 
With $A^x = A^x_{0} - A^x_{1}$ and $B^y = B^y_{0} - B^y_{1}$, where $(A^x_{a})_a$ and $(B^y_{b})_b$ are projective measurements, the approximation bound of Lemma \ref{lemma:two-qubit-robust} translates into a trace distance bound of the actual correlation between Alice, Bob and Eve, 
\[
  \Omega^{xy} := \sum_{a,b} \proj{a}^A \ox \proj{b}^B \ox \tr_{\cA\cB} \Psi (A^x_{a} \ox B^y_{b} \ox \1_{\cE}),
\]
and the ideal one corresponding to attainment of the Tsirelson bound,
\[
  \Omega_0^{xy} := \sum_{a,b} \left(\tr \psi (A^{Qx}_{a} \ox B^{Qy}_{b})\right) \proj{a}^{\cA} \ox \proj{b}^{\cB} \ox \proj{\xi}^{\cE}.
\]
Indeed, for all $x,y\in\{0,1\}$, the theorem implies that if $\langle T_{\theta,\phi,\omega} \rangle \geq \tau(\theta,\phi,\omega)-\epsilon$, there exists a $\ket{\xi} \in \cE$ such that 
\begin{equation}
  \label{eq:correlation}
  \frac12 \left\| \Omega^{xy} - \Omega_0^{xy} \right\|_1 \leq g_2(\epsilon).
\end{equation}
The state $\Omega^{00}$ contains everything needed to evaluate the rate bound in Eq.~\eqref{dw} using Alicki-Fannes continuity bounds for the conditional entropy. The closeness in trace norm implies that the key rate evaluated on the actual state is close to that of the ideal state. Note that the gauge freedom of the self-test, the local unitaries, does not appear anymore as we have traced over $\cA$ and $\cB$.
In detail, we record this insight in the following proposition: 
\begin{proposition}
\label{prop:2x2-DIQKD}
Given a strategy with qubits $\cA$ and $\cB$ for Alice and Bob, such that $\langle T_{\theta,\phi,\omega} \rangle \geq \tau(\theta,\phi,\omega)-\epsilon$ with $\phi$ sufficiently close to $\frac{\pi}{2}$ and $\epsilon\geq 0$ small enough, the key rate of the protocol described in Section \ref{sec:prelim} is 
\begin{equation}\begin{split}
  \label{eq. key rate}
  R &\geq 1 - H\left(\frac{1\pm\sin\phi}{2}\mp g_2(\epsilon)\right) - 2 g_2(\epsilon) - (1+g_2(\epsilon))\, h\!\left( \frac{g_2(\epsilon)}{1+g_2(\epsilon)} \right) \\
    &\geq: 1- \Delta R(\phi) - f_2(\epsilon).
\end{split}\end{equation}
\end{proposition}

\begin{proof}
We use the Alicki-Fannes inequality in the form given in \cite[Lemma~2]{winter2016tight}, which yields that for states as given in Eq.~\eqref{eq:correlation}, and for a sufficiently small $\epsilon$ such that $g_2(\epsilon) \leq 1$,  \begin{equation}
\bigl| H(A|\cE)_{\Omega_0^{00}} - H(A|\cE)_{\Omega^{00}} \bigr|
\leq 2 g_2(\epsilon)  + (1+g_2(\epsilon))\, h\!\left( \frac{g_2(\epsilon)}{1+g_2(\epsilon)} \right).
\end{equation}
Since $\tr_{\cB}\Omega_0^{00} = \frac12\1_A\ox\proj{\xi}^{\cE}$, we have $H(A|E)_{\Omega_0^{00}}=1$, and hence we obtain the bound   
\begin{equation}
  H(A|\cE,X=0)_{\Omega}
     \geq 1 - 2 g_2(\epsilon) - (1+g_2(\epsilon))\, h\!\left( \frac{g_2(\epsilon)}{1+g_2(\epsilon)} \right).
\end{equation}
We can similarly obtain a bound for $H(A|B,X=0,Y=0)_{\Omega}$ as follows. Observe that 
$\Pr\{A=B|\Omega_0^{00}\} = \frac{1+\sin\phi}{2}$, thus by Eq.~\eqref{eq:correlation}, 
\begin{equation}
  \label{eq:2x2-bit-error-rate}
  \Pr\{A=B|X=0,Y=0,\Omega\}
  = \Pr\{A=B|\Omega^{00}\} 
  \geq \frac{1+\sin\phi}{2} - g_2(\epsilon),
\end{equation}
and so by Fano's inequality 
\begin{equation}
  H(A|B,X=0,Y=0)_{\Omega} 
  \leq H\left(\frac{1\pm\sin\phi}{2}\mp g_2(\epsilon)\right).
\end{equation}
Thus we obtain the final bound on the key rate as stated in Eq.~\eqref{eq. key rate}.
\end{proof}

\begin{myproof}[of Theorem~\ref{thm:robust-key-rate}]
For the general case, we return to the consideration at the start of this section, employing Jordan's Lemma, which lead us to Eqs.~\eqref{eq:Jordan-observables} and \eqref{eq:Jordan-state}. The reasoning follows the idea of \cite{Murta:DIQKD}:
as we had observed there, the distribution $p_{ij}$ can for all intents and purposes be considered a real part of the procedure as seen by Eve (though not necessarily Alice and Bob), i.e.~we lose no generality by supposing that the joint state is 
\begin{equation}
\label{eq:Jordan-state-ij}
  \ket{\widetilde{\Psi}}^{\cA\cB\cE IJ} 
    = \sum_{i,j} \sqrt{p_{ij}} \ket{i}^{\cA'}\otimes\ket{j}^{\cB'}\otimes\ket{\widetilde{\Psi}_{ij}}^{\widetilde{\cA}\widetilde{\cB}\cE} \ox \ket{i}^I \ox \ket{j}^J.
\end{equation}
Indeed, after the measurements of $A^x$ and $B^y$, respectively, on the original state $\Psi$, we obtain the classical-quantum state 
\[\begin{split}
  (\Omega^{xy})^{AB\cE} 
    &= \sum_{ab} \proj{a}^A \ox \proj{b}^B \ox \tr_{\cA\cB}\proj{\Psi}(A^x_a\ox B^y_b\ox\1_{\cE}) \\
    &= \sum_{ab} \proj{a}^A \ox \proj{b}^B \ox \sum_{i,j} p_{ij} \tr_{\widetilde{\cA}\widetilde{\cB}}\proj{\widetilde{\Psi}_{ij}}(\tilde{A}^x_{a|i}\ox \tilde{B}^y_{b|j}\ox\1_{\cE}),
\end{split}\]
which is a straight marginal of the following state where $i$ and $j$ are known to Eve: 
\[\begin{split}
  (\widetilde{\Omega}^{xy})^{AB\cE IJ} 
    &= \sum_{ab} \proj{a}^A \ox \proj{b}^B \ox \tr_{\cA\cB}\proj{\widetilde{\Psi}}(A^x_a\ox B^y_b\ox\1_{\cE IJ}) \\
    &= \sum_{ab} \proj{a}^A \ox \proj{b}^B \ox \sum_{i,j} p_{ij} \tr_{\widetilde{\cA}\widetilde{\cB}}\proj{\widetilde{\Psi}_{ij}}(\tilde{A}^x_{a|i}\ox \tilde{B}^y_{b|j}\ox\1_{\cE}) \ox \proj{i}^I \ox \proj{j}^J \\
    &= \sum_{i,j} p_{ij} (\widetilde{\Omega}_{ij}^{xy})^{AB\cE} \ox \proj{i}^I \ox \proj{j}^J,
\end{split}\]
where 
$\widetilde{\Omega}_{ij}^{xy} = \sum_{ab} \proj{a}^A \ox \proj{b}^B \ox \tr_{\widetilde{\cA}\widetilde{\cB}}\proj{\widetilde{\Psi}_{ij}}(\tilde{A}^x_{a|i}\ox \tilde{B}^y_{b|j}\ox\1_{\cE})$. 

As we have observed before, the Bell score $\eta \geq \tau(\theta,\phi,\omega)-\epsilon$ of the original strategy equals the average $\sum_{i,j} p_{ij}\eta_{ij}$ of the Bell scores of the two-qubit strategies $\widetilde{\Psi}_{ij}$ with the measurements $\tilde{A}^x_{|i}$ and $\tilde{B}^y_{|j}$, thus the set $\mathcal{G}$ of pairs $(i,j)$ such that $\eta_{ij} \geq \tau(\theta,\phi,\omega) - \sqrt{\epsilon}$ has probability $\Pr\{(i,j)\in\mathcal{G}\} \geq 1-\sqrt{\epsilon}$. 

For sufficiently small $\epsilon$, we can thus argue for the conditional entropy of the raw key,
\[\begin{split}
  H(A|\cE,X=0)_{\Omega} 
    &\geq H(A|\cE I J,X=0)_{\widetilde{\Omega}} \\
    &= \sum_{i,j} p_{ij} H(A|\cE)_{\widetilde{\Omega}_{ij}} \\
    &\geq \sum_{(i,j)\in\mathcal{G}} p_{ij} H(A|\cE)_{\widetilde{\Omega}_{ij}} \\
    &\geq \left(1-\sqrt{\epsilon}\right)
     \left( 1 - 2 g_2(\sqrt{\epsilon}) - \left(1+g_2(\sqrt{\epsilon})\right)\, h\!\left( \frac{g_2(\sqrt{\epsilon})}{1+g_2(\sqrt{\epsilon})} \right) \right) \\
    &\geq 1 - \sqrt{\epsilon} - 2 g_2(\sqrt{\epsilon}) - \left(1+g_2(\sqrt{\epsilon})\right)\, h\!\left( \frac{g_2(\sqrt{\epsilon})}{1+g_2(\sqrt{\epsilon})} \right),
\end{split}\]
where the first inequality is strong subadditivity (data processing), and the penultimate inequality is from Proposition \ref{prop:2x2-DIQKD}, Eq.~\eqref{eq. key rate}. 

Similarly for the bit error probability, 
\[\begin{split}
  \Pr\{A=B|X=0,Y=0,\Omega\}
    &= \Pr\{A=B|\Omega^{00}\} \\
    &= \Pr\{A=B|\widetilde{\Omega}^{00}\} \\
    &= \sum_{i,j} p_{ij} \Pr\{A=B|\widetilde{\Omega}_{ij}^{00}\} \\
    &\geq \sum_{(i,j)\in\mathcal{G}} p_{ij} \Pr\{A=B|\widetilde{\Omega}_{ij}^{00}\} \\
    &\geq \left(1-\sqrt{\epsilon}\right) \left(\frac{1+\sin\phi}{2} - g_2(\sqrt{\epsilon})\right) \\
    &\geq \frac{1+\sin\phi}{2} - g_2(\sqrt{\epsilon}) - \sqrt{\epsilon},
\end{split}\]
where in the penultimate line we have used Eq.~\eqref{eq:2x2-bit-error-rate} from the proof of Proposition \ref{prop:2x2-DIQKD}. As in that proof, we now use Fano's inequality to upper bound
\[
  H(A|B,X=0,Y=0)_{\Omega} \leq H\left(\frac{1\pm\sin\phi}{2} \mp  g_2({\sqrt{\epsilon}}) \mp \sqrt{\epsilon} \right).
\]

We thus obtain the final lower bound on the key rate as claimed, with 
\begin{equation}
\begin{aligned}
  f(\epsilon) :={}& H\left(\frac{1\pm\sin\phi}{2} \mp g_2(\sqrt{\epsilon}) \mp \sqrt{\epsilon}\right)
     - H\left(\frac{1\pm\sin\phi}{2}\right) \\
  &+ \sqrt{\epsilon} + 2g_2(\sqrt{\epsilon})
     + \bigl(1+g_2(\sqrt{\epsilon})\bigr)\,
       h\!\left(\frac{g_2(\sqrt{\epsilon})}{1+g_2(\sqrt{\epsilon})}\right).
\end{aligned}
\end{equation}

 Since
$g_2(\sqrt{\epsilon})=O(\epsilon^{1/4})$ and $h(x) = O(x\log(1/x))$ as
$x\searrow0$, each term is $O\bigl(\epsilon^{1/4}\log(1/\epsilon)\bigr)$, and
hence $f(\epsilon) = O\bigl(\epsilon^{1/4}\log(1/\epsilon)\bigr)$, as claimed.
\end{myproof}

\subsection{Robust self-test from near-maximum Bell violation}

\label{sec:general-robust}

We now lift the two-qubit robust self-tests of
Subsections~\ref{two-qubit-robust:compactness} and~\ref{two qubit strategy}
to strategies of arbitrary dimension. In the two-qubit self-test, we have, for the set $\mathcal{G}$  of pairs $(i,j)$ such that $\eta_{ij} \geq \tau(\theta,\phi,\omega) - \sqrt{\epsilon}$, there exist local qubit unitaries $U^{ij}$ and $V^{ij}$
that bring the block observables close to the ideal quantum observables and the
block state close to the ideal quantum state. However, since Alice's unitary $U^{ij}$ depends on Bob's block index $j$, and similarly for Bob, these unitaries cannot be combined directly into local isometries on the full Hilbert
space. Hence, we will now prove that for fixed $i$, the unitaries
$U^{ij}$ and $U^{ij'}$ coincide up to a global phase and a small error and that the analogous statement holds
for $V^{ij}$ and $V^{i'j}$ at fixed $j$. We state the following lemma for a single party, in Alice’s notation, but it also applies to Bob verbatim.

\begin{lemma}
\label{lem:phase-alignment}
Let $\tilde{A}^0,\tilde{A}^1$ and $A^{Q0},A^{Q1}$ be $\pm1$ qubit observables,
and let $U$ and $U'$ be qubit unitaries such that
\begin{equation}
\label{pa:block}
  \norm{U\tilde{A}^xU^\dagger - A^{Qx}} \le \mathcal{F}_x, \qquad
  \norm{U'\tilde{A}^xU'^\dagger - A^{Qx}} \le \mathcal{F}_x
  \qquad\text{for } x\in\{0,1\},
\end{equation}
and for some $\mathcal{F}_0,\mathcal{F}_1\geq0$. Let the
eigenbases $\{\ket{e_k}\}$ of $\tilde{A}^0$ and $\{\ket{f_l}\}$ of
$\tilde{A}^1$ satisfy
\begin{equation}
\label{pa:overlap}
  \lvert\braket{e_k|f_l}\rvert^2\ge\delta
  \qquad\text{for all } k,l\in\{0,1\},
\end{equation}
for some $\delta>0$. Then there is a phase $e^{\mathrm{i}\gamma}$ with
\begin{equation}
\label{pa:bound}
  \norm{U' - e^{\mathrm{i}\gamma}U}
  \le \mathcal{F}_0 + \mathcal{F}_0^2
     + \frac{2(\mathcal{F}_0+\mathcal{F}_1)}{\sqrt{\delta}}
  =: a(\mathcal{F}_0,\mathcal{F}_1,\delta).
\end{equation}\end{lemma}

\begin{proof}
Let us set
\begin{equation}\label{eq:W}
  W := U^\dagger U'.
\end{equation}
Since $U' = UW$ and left multiplication by a unitary preserves the operator
norm, $\norm{U' - e^{\mathrm{i}\gamma}U} = \norm{W - e^{\mathrm{i}\gamma}\1}$
for every phase $e^{\mathrm{i}\gamma}$, so it suffices to bound the latter.
If $\mathcal{F}_0\geq1$, the right-hand side of~\eqref{pa:bound} is at least
$2$ which is greater than or equal to $\norm{W - e^{\mathrm{i}\gamma}\1}$ and the bound is trivial, so we
assume $\mathcal{F}_0<1$.

Let $\{\ket{e_k}\}_{k=0,1}$ be the eigenbasis of $\tilde{A}^0$ and
$\{\ket{f_l}\}_{l=0,1}$ that of $\tilde{A}^1$, and let
\begin{equation}\label{eq:polar}
  D_0 := \sum_{k} d_k\proj{e_k}, \quad d_k := \bra{e_k}W\ket{e_k} = r_ke^{\mathrm{i}\alpha_k},
  \qquad
  D_1 := \sum_{l} \lambda_l\proj{f_l}, \quad \lambda_l := \bra{f_l}W\ket{f_l},
\end{equation}
with $r_k\geq0$, be the diagonal parts of $W$ in these bases. We choose the
phase $e^{\mathrm{i}\gamma} := e^{\mathrm{i}\alpha_0}$; we show below that
$r_0>0$, so that this phase is well defined. Using the triangle inequality,
we have
\begin{equation}
  \norm{W - e^{\mathrm{i}\gamma}\1} \le \norm{W - D_0}+\norm{D_0-e^{\mathrm{i}\gamma}\1}.
\end{equation}

Let us first obtain the bound on $\norm{W - D_0}$. Writing
\[
  U\tilde{A}^xU^\dagger = A^{Qx}+E_x,
  \qquad
  U'\tilde{A}^xU'^\dagger = A^{Qx}+E'_x,
\]
with $\norm{E_x},\norm{E'_x}\le\mathcal{F}_x$ by~\eqref{pa:block}, the first
relation gives $\tilde{A}^x = U^\dagger(A^{Qx}+E_x)U$, and therefore
\[
  W\tilde{A}^xW^\dagger
  = U^\dagger\bigl(U'\tilde{A}^xU'^\dagger\bigr)U
  = U^\dagger(A^{Qx}+E'_x)U
  = \tilde{A}^x + U^\dagger(E'_x-E_x)U.
\]
Since $[W,\tilde{A}^x] = \bigl(W\tilde{A}^xW^\dagger - \tilde{A}^x\bigr)W$,
and both $W$ and $U$ are unitary,
\begin{equation}
\label{eq:commutator}
  \norm{[W,\tilde{A}^x]}
  = \norm{W\tilde{A}^xW^\dagger - \tilde{A}^x}
  = \norm{E'_x-E_x}
  \le 2\mathcal{F}_x.
\end{equation}

We work in the eigenbasis $\{\ket{e_0},\ket{e_1}\}$ of $\tilde{A}^0$, in
which $\tilde{A}^0 = \mathrm{diag}(1,-1)$, and write
\begin{equation}
  W
  = \begin{pmatrix} d_0 & x_{01} \\ x_{10} & d_1 \end{pmatrix}
  = \underbrace{\begin{pmatrix} d_0 & 0 \\ 0 & d_1 \end{pmatrix}}_{D_0}
  + \underbrace{\begin{pmatrix} 0 & x_{01} \\ x_{10} & 0 \end{pmatrix}}_{X},
\end{equation}
so that $W - D_0 = X$ is the off-diagonal part of $W$, and it suffices to bound
$\norm{X}$. Since $D_0$ and $\tilde{A}^0$ are both diagonal in this basis,
they commute, and hence
\begin{equation}
  \bigl[W,\tilde{A}^0\bigr] = \bigl[X,\tilde{A}^0\bigr].
\end{equation}
Explicitly, this gives
\begin{equation}
  \bigl[X,\tilde{A}^0\bigr]
  = \begin{pmatrix} 0 & -2x_{01} \\ 2x_{10} & 0 \end{pmatrix}.
\end{equation}
Since $X$ is antidiagonal, we have
$\norm{X} = \max\bigl(\lvert x_{01}\rvert,\lvert x_{10}\rvert\bigr)$. Hence
\begin{equation}
  \norm{\bigl[X,\tilde{A}^0\bigr]}
  = 2\max\bigl(\lvert x_{01}\rvert,\lvert x_{10}\rvert\bigr)
  = 2\norm{X}.
\end{equation}
Combining these identities with the commutator bound~\eqref{eq:commutator}
for $x=0$, we obtain
\begin{equation}
\label{eq:W-D0}
  \norm{W - D_0}
  = \norm{X}
  = \tfrac12\norm{\bigl[X,\tilde{A}^0\bigr]}
  = \tfrac12\norm{\bigl[W,\tilde{A}^0\bigr]}
  \le \mathcal{F}_0.
\end{equation}

Now we bound $\norm{D_0-e^{\mathrm{i}\gamma}\1}$. The operator
$D_0 - e^{\mathrm{i}\gamma}\1 = \sum_k (d_k - e^{\mathrm{i}\gamma})\proj{e_k}$
is diagonal in an orthonormal basis, so its operator norm equals its largest
eigenvalue in modulus:
\begin{equation}
\label{eq:D0-max}
  \norm{D_0 - e^{\mathrm{i}\gamma}\1}
  = \max_{k\in\{0,1\}}\,\lvert d_k - e^{\mathrm{i}\gamma}\rvert .
\end{equation}

First we look at $k=0$. Since $W$ is unitary, its first column in the basis
$\{\ket{e_0},\ket{e_1}\}$ has unit norm, so
$r_0^2 + \lvert x_{10}\rvert^2 = 1$. This gives $r_0\le1$, and
$1 - r_0^2 = \lvert x_{10}\rvert^2 \le \norm{X}^2 = \norm{W - D_0}^2
 \le \mathcal{F}_0^2$ by~\eqref{eq:W-D0}. In particular, $r_0>0$ since
$\mathcal{F}_0<1$, so the phase $e^{\mathrm{i}\gamma}=e^{\mathrm{i}\alpha_0}$ is
well defined. Since $d_0 = r_0e^{\mathrm{i}\gamma}$, we have
$d_0 - e^{\mathrm{i}\gamma} = (r_0 - 1)\,e^{\mathrm{i}\gamma}$, and thus
\begin{equation}
\label{eq:modulus}
  \lvert d_0 - e^{\mathrm{i}\gamma}\rvert
  = 1 - r_0 \le (1 - r_0)(1 + r_0) = 1 - r_0^2 \le \mathcal{F}_0^2 .
\end{equation}

Next we look at $k=1$. The same argument that led to~\eqref{eq:W-D0}, applied
to $\tilde{A}^1$, gives $\norm{W - D_1}\le\mathcal{F}_1$, and therefore
\begin{equation}
\label{eq:D0-D1}
  \norm{D_0 - D_1} \le \norm{D_0 - W} + \norm{W - D_1}
  \le \mathcal{F}_0 + \mathcal{F}_1 .
\end{equation}
Since $\bra{e_k}D_0 = d_k\bra{e_k}$ and $D_1\ket{f_l} = \lambda_l\ket{f_l}$,
\begin{equation}
  \bra{e_k}(D_0 - D_1)\ket{f_l} = (d_k - \lambda_l)\,\braket{e_k|f_l} .
\end{equation}
The left-hand side is a matrix element between unit vectors, so its modulus is
at most $\norm{D_0 - D_1}\le \mathcal{F}_0 + \mathcal{F}_1$. Dividing by
$\lvert\braket{e_k|f_l}\rvert\ge\sqrt{\delta}$, using~\eqref{pa:overlap},
gives $\lvert d_k - \lambda_l\rvert\le(\mathcal{F}_0+\mathcal{F}_1)/\sqrt{\delta}$
for all $k,l$. Routing through any fixed $l$,
\begin{equation}
\label{eq:d0-d1}
  \lvert d_1 - d_0\rvert
  \le \lvert d_1 - \lambda_l\rvert + \lvert\lambda_l - d_0\rvert
  \le \frac{2(\mathcal{F}_0 + \mathcal{F}_1)}{\sqrt{\delta}} ,
\end{equation}
and hence, by~\eqref{eq:modulus},
\begin{equation}
  \lvert d_1 - e^{\mathrm{i}\gamma}\rvert
  \le \lvert d_1 - d_0\rvert + \lvert d_0 - e^{\mathrm{i}\gamma}\rvert
  \le \frac{2(\mathcal{F}_0 + \mathcal{F}_1)}{\sqrt{\delta}} + \mathcal{F}_0^2 .
\end{equation}

The bound for $k=1$ is larger than the one for $k=0$, so by~\eqref{eq:D0-max},
$\norm{D_0 - e^{\mathrm{i}\gamma}\1}\le 2(\mathcal{F}_0+\mathcal{F}_1)/\sqrt{\delta}
 + \mathcal{F}_0^2$. Together with~\eqref{eq:W-D0},
\begin{equation}
  \norm{W - e^{\mathrm{i}\gamma}\1}
  \le \norm{W - D_0} + \norm{D_0 - e^{\mathrm{i}\gamma}\1}
  \le \mathcal{F}_0 + \mathcal{F}_0^2
      + \frac{2(\mathcal{F}_0 + \mathcal{F}_1)}{\sqrt{\delta}}
  = a(\mathcal{F}_0,\mathcal{F}_1,\delta).
\end{equation}
\end{proof}

Using Lemma~\ref{lem:phase-alignment}, we can select a single reference unitary $U^i$ for
each $i$ and $V^j$ for each $j$, which together define local unitaries on
Alice's and Bob's full Hilbert spaces. The residual phases which depend on the block can be absorbed into the junk vector of that block. 
This leads to the proof of Theorem~\ref{thm:robust-self-test-general}.

\begin{myproof}[of Theorem~\ref{thm:robust-self-test-general}]
We identify $\cA\cong\cA'\ox\widetilde{\cA}$ and $\cB\cong\cB'\ox\widetilde{\cB}$
as in Subsection~\ref{jordan lemma}, with $\widetilde{\cA}$ and
$\widetilde{\cB}$ qubits, and recall that the set $\mathcal{G}$ of blocks
$(i,j)$ with $\eta_{ij}\geq\tau(\theta,\phi,\omega)-\sqrt{\epsilon}$ has total
weight $\sum_{(i,j)\in\mathcal{G}}p_{ij}\geq1-\sqrt{\epsilon}$.

We look at the good blocks first. Let $(i,j)\in\mathcal{G}$, and denote by
$S_{ij}$ the two-qubit strategy of this block, with observables
$\tilde{A}^x_{|i}=\tilde{A}^x_{0|i}-\tilde{A}^x_{1|i}$ and
$\tilde{B}^y_{|j}=\tilde{B}^y_{0|j}-\tilde{B}^y_{1|j}$. Its Bell score $\eta_{ij}$ is greater than or equal to $\tau(\theta,\phi,\omega)-\sqrt{\epsilon}$, so Lemma~\ref{lemma:two-qubit-robust} provides qubit unitaries $U^{ij}$
and $V^{ij}$ and a unit vector $\ket{\xi'_{ij}}^{\cE}$ such that
\begin{equation}
\label{eq:block-g2}
  \norm{(U^{ij}\ox V^{ij}\ox\1_{\cE})
        (\tilde{A}^x_{|i}\ox\tilde{B}^y_{|j}\ox\1_{\cE})
        \ket{\widetilde{\Psi}_{ij}}^{\widetilde{\cA}\widetilde{\cB}\cE} 
        - (A^{Qx}\ox B^{Qy})\ket{\Phi_0}^{\cA \cB} \ox\ket{\xi'_{ij}}^{\cE} }_2
  \le g_2(\sqrt{\epsilon})
\end{equation}
for all $x,y$. The same unitaries also bound the observables. By the
proof of Lemma~\ref{lemma:two-qubit-robust}, they satisfy
$d\bigl(S_{ij},S_0^{(U^{ij})^\dagger,(V^{ij})^\dagger}\bigr)
 \leq\delta(\sqrt{\epsilon})$, where $\delta$ is the function of
Lemma~\ref{qubit strategy}, which we may take to be monotonic, so that a
single value $\delta(\sqrt{\epsilon})$ applies to every good block. Since the
distance $d$ is a sum of non-negative terms, this implies
$\norm{\tilde{A}^x_{a|i} - (U^{ij})^\dagger A^{Qx}_aU^{ij}}
 \leq\delta(\sqrt{\epsilon})$ for each $x$ and $a$, and the triangle
inequality applied to $\tilde{A}^x_{|i}=\tilde{A}^x_{0|i}-\tilde{A}^x_{1|i}$
and $A^{Qx}=A^{Qx}_0-A^{Qx}_1$ gives
\begin{equation}
  \norm{U^{ij}\tilde{A}^x_{|i}(U^{ij})^\dagger - A^{Qx}}
  = \norm{\tilde{A}^x_{|i} - (U^{ij})^\dagger A^{Qx}U^{ij}}
  \leq 2\delta(\sqrt{\epsilon}),
\end{equation}
where the equality follows from unitary invariance of the operator norm. The same argument applies to Bob's observables with $V^{ij}$.

We denote these bounds by
\begin{equation}
\label{eq:F-compactness}
  \mathcal{F}_{A_0}=\mathcal{F}_{A_1}=\mathcal{F}_{B_0}=\mathcal{F}_{B_1}
  := 2\delta(\sqrt{\epsilon}),
\end{equation}
so that, for every good block,
$\norm{U^{ij}\tilde{A}^x_{|i}(U^{ij})^\dagger - A^{Qx}}\leq\mathcal{F}_{A_x}$
and
$\norm{V^{ij}\tilde{B}^y_{|j}(V^{ij})^\dagger - B^{Qy}}\leq\mathcal{F}_{B_y}$.

Next, we show that the eigenbases of each party's two observables are not orthogonal, so that the overlap condition given in inequality~\eqref{pa:overlap} of
Lemma~\ref{lem:phase-alignment} holds. Writing a $\pm1$ qubit observable as
$\vec{n}\cdot\vec{\sigma}$, two such observables with Bloch vectors
$\vec{n}_0,\vec{n}_1$ have eigenbases satisfying
$\min_{k,l}\lvert\braket{e_k|f_l}\rvert^2
 = \tfrac12\bigl(1-\lvert\vec{n}_0\cdot\vec{n}_1\rvert\bigr)$, with
$\vec{n}_0\cdot\vec{n}_1
 = \tfrac12\tr\bigl[(\vec{n}_0\cdot\vec{\sigma})(\vec{n}_1\cdot\vec{\sigma})\bigr]$.
For Alice's observables in a good block,
\begin{equation}
\begin{aligned}
  &\tr\bigl[\tilde{A}^0_{|i}\tilde{A}^1_{|i}\bigr] - \tr\bigl[A^{Q0}A^{Q1}\bigr] \\
  &\qquad= \tr\Bigl[\bigl(U^{ij}\tilde{A}^0_{|i}(U^{ij})^\dagger - A^{Q0}\bigr)
              U^{ij}\tilde{A}^1_{|i}(U^{ij})^\dagger\Bigr]
   + \tr\Bigl[A^{Q0}\bigl(U^{ij}\tilde{A}^1_{|i}(U^{ij})^\dagger - A^{Q1}\bigr)\Bigr],
\end{aligned}
\end{equation}
so that, using $\lvert\tr[MN]\rvert\le\norm{M}\,\norm{N}_1$,
$\norm{N}_1=2$ for $\pm1$ qubit observables, and
$\tfrac12\tr[A^{Q0}A^{Q1}]=\cos\theta$, we obtain
$\lvert\vec{n}_0\cdot\vec{n}_1-\cos\theta\rvert
 \le\mathcal{F}_{A_0}+\mathcal{F}_{A_1}$. In the same way, since
$\tfrac12\tr[B^{Q0}B^{Q1}]=\cos(\omega-\phi)$, Bob's observables satisfy
$\lvert\vec{n}_0\cdot\vec{n}_1-\cos(\omega-\phi)\rvert
 \le\mathcal{F}_{B_0}+\mathcal{F}_{B_1}$. Hence the overlap conditions of
Lemma~\ref{lem:phase-alignment} hold for every good block with
\begin{equation}
\label{eq:delta-AB}
  \delta_A := \tfrac12\bigl(1-\lvert\cos\theta\rvert
               -\mathcal{F}_{A_0}-\mathcal{F}_{A_1}\bigr),
  \qquad
  \delta_B := \tfrac12\bigl(1-\lvert\cos(\omega-\phi)\rvert
               -\mathcal{F}_{B_0}-\mathcal{F}_{B_1}\bigr).
\end{equation}

By Eq.~\eqref{eq: condition for quantum greater than classical},
$\tau(\theta,\phi,\omega)>\beta(\theta,\phi,\omega)\geq0$, so neither
$\sin\theta$ nor $\sin(\omega-\phi)$ vanishes, i.e.\ $\lvert\cos\theta\rvert<1$
and $\lvert\cos(\omega-\phi)\rvert<1$.  Since the $\mathcal{F}$'s tend to $0$
with $\epsilon$, we obtain  $\delta_A,\delta_B>0$ for
a sufficiently small $\epsilon$.

The unitaries $U^{ij}$ and $V^{ij}$ depend on both block indices, whereas a
local unitary for Alice may only depend on $i$, and one for Bob only on $j$.
We therefore choose a reference block for each index: for each $i$ such that
$(i,j)\in\mathcal{G}$ for some $j$, we fix one such reference block $(i,j')$,
where $j'$ depends on $i$, and for each $j$ such that $(i,j)\in\mathcal{G}$
for some $i$, we fix one such reference block $(i',j)$, where $i'$ depends
on $j$. We set
\begin{equation}
  U^i := U^{ij'}, \qquad V^j := V^{i'j},
\end{equation}
while $U^i := \1$ and $V^j := \1$ for all indices $i$ and $j$ that occur in
no good block. Now let $(i,j)\in\mathcal{G}$. Alice's observables are the
same, $\tilde{A}^x_{|i}$, in the block $(i,j)$ and in the reference block
$(i,j')$, and by the bounds on the observables derived above, both $U^{ij}$
and $U^i = U^{ij'}$ map them close to $A^{Qx}$; moreover,
by~\eqref{eq:delta-AB}, their eigenbases have overlaps at least $\delta_A$,
which is positive for a sufficiently small $\epsilon$. Hence
Lemma~\ref{lem:phase-alignment}, applied with $U := U^{ij}$, $U' := U^{ij'}$,
$\mathcal{F}_x := \mathcal{F}_{A_x}$ and $\delta := \delta_A$, yields a phase
$\gamma\in\mathbb{R}$ with $\norm{U^i - e^{\mathrm{i}\gamma}U^{ij}} \le a_A$.
In the same way, applied to Bob's observables $\tilde{B}^y_{|j}$ in the
blocks $(i,j)$ and $(i',j)$, it yields a phase $\mu\in\mathbb{R}$ with
$\norm{V^j - e^{\mathrm{i}\mu}V^{ij}} \le a_B$. Thus, for every
$(i,j)\in\mathcal{G}$,
\begin{equation}
\label{eq:term1-alignment}
  \norm{U^i - e^{\mathrm{i}\gamma}U^{ij}} \le a_A,
  \qquad
  \norm{V^j - e^{\mathrm{i}\mu}V^{ij}} \le a_B,
\end{equation}
with phases $\gamma$ and $\mu$ depending on the block, where
$a_A := a(\mathcal{F}_{A_0},\mathcal{F}_{A_1},\delta_A)$ and
$a_B := a(\mathcal{F}_{B_0},\mathcal{F}_{B_1},\delta_B)$.

Now we assemble the bounds of each block. Define the block-diagonal
unitaries
\begin{equation}
  U:=\sum_i\proj{i}\ox U^i,
  \qquad
  V:=\sum_j\proj{j}\ox V^j,
\end{equation}
which are local, since each depends only on that party's own block index.

For $(i,j)\in\mathcal{G}$, let $\gamma$ and $\mu$ be the phases
in~\eqref{eq:term1-alignment} for this block, and define
$\ket{\xi_{ij}}^{\cE} := e^{\mathrm{i}(\gamma+\mu)}\ket{\xi'_{ij}}^{\cE}$ and for
$(i,j)\notin\mathcal{G}$, let $\ket{\xi_{ij}}^{\cE}$ be an arbitrary unit
vector. We set
\begin{equation}
  \ket{\xi}^{\cA'\cB'\cE}
  := \sum_{i,j}\sqrt{p_{ij}}\,\ket{i}^{\cA'}\ox\ket{j}^{\cB'}\ox\ket{\xi_{ij}}^{\cE},
\end{equation}
which is a unit vector since $\sum_{i,j}p_{ij}=1$. We now decompose the error into blocks.
\begin{multline}
  (U\ox V\ox\1_{\cE})(A^x\ox B^y\ox\1_{\cE})\ket{\Psi}^{\cA\cB\cE}
  - (A^{Qx}\ox B^{Qy})\ket{\Phi_0}^{\widetilde{\cA}\widetilde{\cB}}
    \ox\ket{\xi}^{\cA'\cB'\cE} \\
  = \sum_{i,j}\sqrt{p_{ij}}\;\ket{i}^{\cA'}\ox\ket{j}^{\cB'}
    \ox\ket{\Delta_{ij}}^{\widetilde{\cA}\widetilde{\cB}\cE},
\end{multline}
where
\begin{equation}
  \ket{\Delta_{ij}}^{\widetilde{\cA}\widetilde{\cB}\cE}
  := (U^i\ox V^j\ox\1_{\cE})(\tilde{A}^x_{|i}\ox\tilde{B}^y_{|j}\ox\1_{\cE})
     \ket{\widetilde{\Psi}_{ij}}^{\widetilde{\cA}\widetilde{\cB}\cE}
     - (A^{Qx}\ox B^{Qy})\ket{\Phi_0}^{\widetilde{\cA}\widetilde{\cB}}
       \ox\ket{\xi_{ij}}^{\cE} .
\end{equation}

For $(i,j)\in\mathcal{G}$, add and subtract
$e^{\mathrm{i}(\gamma+\mu)}(U^{ij}\ox V^{ij}\ox\1_{\cE})
 (\tilde{A}^x_{|i}\ox\tilde{B}^y_{|j}\ox\1_{\cE})\ket{\widetilde{\Psi}_{ij}}^{\widetilde{\cA}\widetilde{\cB}\cE}$ to obtain,
\begin{equation}
\begin{aligned}
  \norm{\ket{\Delta_{ij}}}_2
  &\le \underbrace{\Bigl\|(U^i\ox V^j\ox\1_{\cE})
        (\tilde{A}^x_{|i}\ox\tilde{B}^y_{|j}\ox\1_{\cE})
        \ket{\widetilde{\Psi}_{ij}}
      - e^{\mathrm{i}(\gamma+\mu)}(U^{ij}\ox V^{ij}\ox\1_{\cE})
        (\tilde{A}^x_{|i}\ox\tilde{B}^y_{|j}\ox\1_{\cE})
        \ket{\widetilde{\Psi}_{ij}}\Bigr\|_2}_{\text{Term 1}} \\
  &\quad
   + \underbrace{\Bigl\|e^{\mathrm{i}(\gamma+\mu)}(U^{ij}\ox V^{ij}\ox\1_{\cE})
        (\tilde{A}^x_{|i}\ox\tilde{B}^y_{|j}\ox\1_{\cE})
        \ket{\widetilde{\Psi}_{ij}}
      - (A^{Qx}\ox B^{Qy})\ket{\Phi_0}\ox\ket{\xi_{ij}}\Bigr\|_2}_{\text{Term 2}} .
\end{aligned}
\end{equation}

 We bound both the terms separately. 

\emph{Term 1.}
Let $\ket{\chi_{ij}}^{\widetilde{\cA}\widetilde{\cB}\cE}
:= (\tilde{A}^x_{|i}\ox\tilde{B}^y_{|j}\ox\1_{\cE})
\ket{\widetilde{\Psi}_{ij}}^{\widetilde{\cA}\widetilde{\cB}\cE}$. Since
$\tilde{A}^x_{|i}$ and $\tilde{B}^y_{|j}$ are $\pm1$ observables, they are
Hermitian and square to the identity, hence unitary, so
$\norm{\ket{\chi_{ij}}}_2 = \norm{\ket{\widetilde{\Psi}_{ij}}}_2 = 1$.
Factoring out $\ket{\chi_{ij}}$, and using that the operator norm bounds the
action on unit vectors and that $\norm{X\ox\1_{\cE}} = \norm{X}$, we obtain
\begin{equation}
\label{eq:term1-factor}
  \text{Term 1}
  \le \norm{U^i\ox V^j - e^{\mathrm{i}(\gamma+\mu)}\,U^{ij}\ox V^{ij}} .
\end{equation}
Adding and subtracting $e^{\mathrm{i}\gamma}U^{ij}\ox V^j$, and using that
unitaries have operator norm one,
\begin{equation}
\begin{aligned}
  \norm{U^i\ox V^j - e^{\mathrm{i}(\gamma+\mu)}\,U^{ij}\ox V^{ij}}
  &= \norm{\bigl(U^i - e^{\mathrm{i}\gamma}U^{ij}\bigr)\ox V^j
          + e^{\mathrm{i}\gamma}U^{ij}\ox\bigl(V^j - e^{\mathrm{i}\mu}V^{ij}\bigr)} \\
  &\le \norm{U^i - e^{\mathrm{i}\gamma}U^{ij}}
     + \norm{V^j - e^{\mathrm{i}\mu}V^{ij}}
  \le a_A + a_B
\end{aligned}
\end{equation}
by~\eqref{eq:term1-alignment}. Together with~\eqref{eq:term1-factor}, this gives
\begin{equation}
\label{eq:term1-bound}
  \text{Term 1} \le a_A + a_B .
\end{equation}

\emph{Term 2.} Since $\ket{\xi_{ij}} = e^{\mathrm{i}(\gamma+\mu)}\ket{\xi'_{ij}}$,
both vectors carry the common phase $e^{\mathrm{i}(\gamma+\mu)}$, which factors
out:
\begin{equation}
\begin{aligned}
  \text{Term 2}
  &= \Bigl\|e^{\mathrm{i}(\gamma+\mu)}\Bigl((U^{ij}\ox V^{ij}\ox\1_{\cE})
        (\tilde{A}^x_{|i}\ox\tilde{B}^y_{|j}\ox\1_{\cE})
        \ket{\widetilde{\Psi}_{ij}}
      - (A^{Qx}\ox B^{Qy})\ket{\Phi_0}\ox\ket{\xi'_{ij}}\Bigr)\Bigr\|_2 \\
  &= \Bigl\|(U^{ij}\ox V^{ij}\ox\1_{\cE})
        (\tilde{A}^x_{|i}\ox\tilde{B}^y_{|j}\ox\1_{\cE})
        \ket{\widetilde{\Psi}_{ij}}
      - (A^{Qx}\ox B^{Qy})\ket{\Phi_0}\ox\ket{\xi'_{ij}}\Bigr\|_2
  \le g_2(\sqrt{\epsilon}),
\end{aligned}
\end{equation}
by~\eqref{eq:block-g2}. Hence
$\norm{\ket{\Delta_{ij}}}_2\le g_2(\sqrt{\epsilon}) + a_A + a_B$ for
$(i,j)\in\mathcal{G}$.

For $(i,j)\notin\mathcal{G}$,
$\ket{\Delta_{ij}}$ is the difference of two unit vectors, so
$\norm{\ket{\Delta_{ij}}}_2\le2$, and these blocks carry total weight
$\sum_{(i,j)\notin\mathcal{G}}p_{ij}\le\sqrt{\epsilon}$. Since the vectors
$\ket{i}^{\cA'}\ox\ket{j}^{\cB'}$ are orthonormal,
\begin{equation}
\begin{aligned}
  &\norm{(U\ox V\ox\1_{\cE})(A^x\ox B^y\ox\1_{\cE})\ket{\Psi}
        - (A^{Qx}\ox B^{Qy})\ket{\Phi_0}\ox\ket{\xi}}_2^2
  = \sum_{i,j}p_{ij}\norm{\ket{\Delta_{ij}}}_2^2 \\
  &\qquad\le \bigl(g_2(\sqrt{\epsilon}) + a_A + a_B\bigr)^2
      + 4\sqrt{\epsilon} .
\end{aligned}
\end{equation}
Taking the square root gives us\begin{equation*}
    \left\| \bigl( U\otimes V \otimes \1_{\cE}\bigr)
            \bigl( A^x\otimes B^y \otimes \1_{\cE} \bigr)\ket{\Psi}^{\cA\cB\cE} 
       - \bigl( A^{Qx}\otimes B^{Qy} \otimes \1_{\cE}\bigr) \ket{\Phi_0}^{\cA\cB}
            \otimes \ket{\xi}^{\cE} \right\|_2
    \leq g(\epsilon),
\end{equation*}
with $
  g(\epsilon) := \sqrt{\bigl(g_2(\sqrt{\epsilon}) + a_A + a_B\bigr)^2
     + 4\sqrt{\epsilon}} $.
Finally, $g(\epsilon)\to0$ as $\epsilon\searrow0$, since
$g_2(\sqrt{\epsilon})\to0$, while
$\delta_A$ and $\delta_B$ tend to the positive constants
$\frac12(1-\lvert\cos\theta\rvert)$ and $\frac12(1-\lvert\cos(\omega-\phi)\rvert)$ respectively,
so that $a_A,a_B\to0$. Since $g_2(\sqrt{\epsilon})$, $a_A(\sqrt{\epsilon})$, and $a_B(\sqrt{\epsilon})$ scale as $O(\epsilon^{1/4})$, all terms of $g(\epsilon)$, including the contribution
$2\epsilon^{1/4}$ of the bad blocks, are $O(\epsilon^{1/4})$, and hence
$g(\epsilon) = O(\epsilon^{1/4})$, with an implicit constant depending only on
$\theta$, $\phi$ and $\omega$.

\end{myproof}

\section{Numerical results}
\label{numericals}
In the previous section, we have proved that for a Bell violation $\eta\geq\tau(\theta,\phi,\omega)-\epsilon$ close enough to the maximum, we obtain device-independent key rates arbitrarily close to $1-\Delta R(\phi)$. When we stay in the regime of $\phi \approx \frac{\pi}{2}$, the key rate is dominated by $H(A|\cE,X=0)_{\Omega} \approx 1$.
In this section, we give explicit numerical examples of lower bounds on $H(A|\cE,X=0)_{\Omega}$ using established techniques of semidefinite relaxation of the noncommutative optimisation problem.
 
\begin{figure}[ht]
    \centering
    \includegraphics[width=\linewidth]{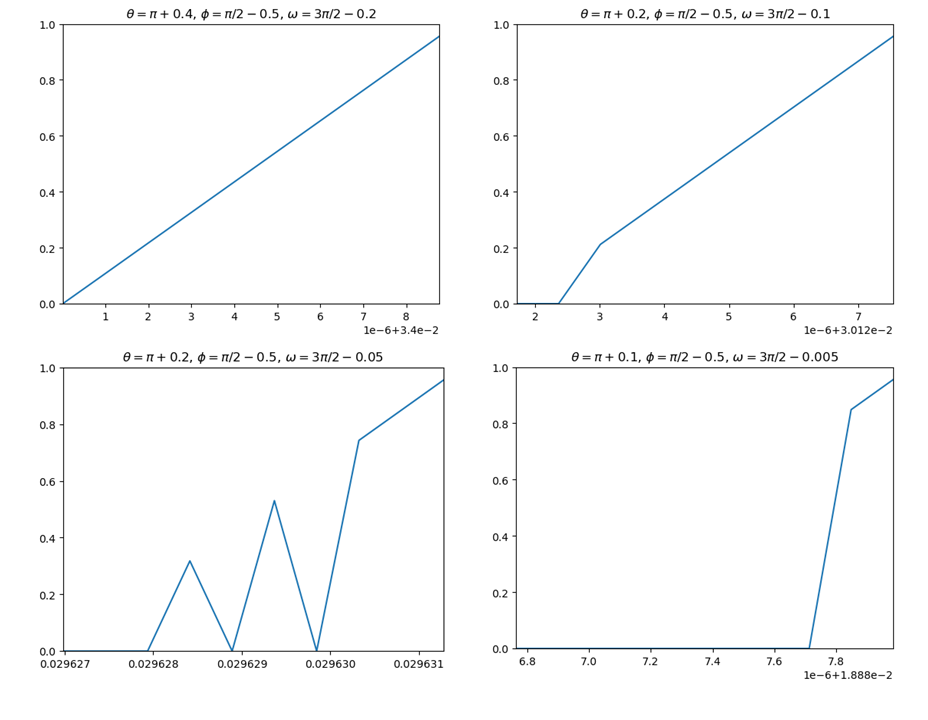}
    \caption{
    Lower bound on the conditional entropy $H(A|\cE,X=0)$ versus observed Bell score $\eta \in (\beta(\theta,\phi,\omega), \tau(\theta,\phi,\omega))$ for the Bell correlator $T_{\theta,\phi,\omega}$ \cite{Wooltorton_2024}. The chosen parameters $(\theta,\phi,\omega)$ move the quantum correlations arbitrarily close to the local boundary while maintaining a CHSH violation $\gtrsim 2$ and $c_0,c_1>0$. The numerical SDP relaxations complement the analytical results by detailing key rates under sub-maximal violations ($\eta < \tau$): while positive key distillation remains theoretically robust near maximal violations ($\eta \to \tau$), the curves steepen sharply as local limits are approached, showing that allowable experimental noise shrinks rapidly.
    }
    \label{numericals: wooltorton}
\end{figure}


In the most general setting, $\Omega^{00}$ is a classical-quantum state of the form $(\Omega^{00})^{A\cE} = \sum_a |a\rangle\langle a| \otimes \tr_{\cA\cB}\!\left[(A_a^{0} \otimes \1_{\cB\cE})\proj{\Psi}^{\cA\cB\cE}\right]$,
where $\{A_a^{0}\}_{a=0}^{n-1}$ is the basis used by Alice to generate the secret key. Then a bound can be obtained by writing $H(A|X=0,E)_{\Omega}$ in terms of the relative entropy 
\begin{equation}
  H(A|\cE,X=0)_{\Omega} = -D(\rho^{A\cE} \| \1_A \otimes \Omega^{\cE}),
\end{equation}
where $D(\rho \| \sigma) = \tr \rho(\log \rho - \log \sigma)$ is the quantum relative entropy, and maximizing it over the Hilbert spaces of Alice, Bob, and Eve, the state $\ket{\Psi}^{\cA\cB\cE}$, the measurements of Alice and Bob that generate the key and violate the Bell inequality, and all the operations of Eve. In general, such an optimization is hard and not feasible. In \cite{Brown_2021}, this problem is broken down into a sequence of optimization problems that converge to $H(A|\cE,X=0)$. This sequence of problems can then be relaxed into SDPs using the NPA hierarchy \cite{PhysRevLett.98.010401,Pironio_2010}.   
 
We apply this method to get lower bounds on $H(A|\cE,X=0)$ for violations  $\eta\geq\tau(\theta,\phi,\omega)-\epsilon$ and plot the results in Fig.~\ref{numericals: wooltorton}. We chose suitable parameters $(\theta,\phi,\omega)$ that take us close to the local set to give a CHSH violation $\gtrsim2$ whilst keeping $c_0,c_1>0$.   
For this choice of parameters, we can plot $H(A|\cE,X=0)$ versus the Bell violation from $\beta$ to $\tau$ in Fig. \ref{numericals: wooltorton}.

We see that while quantum correlations arbitrarily close to the local set can theoretically guarantee secret key distillation, their practical noise tolerance drops precipitously. This adds on to our analytical results. As the system shifts closer to local behaviors, the entropy curves move drastically rightward and steepen sharply. Consequently, certifying secret key bits near the classical-quantum boundary demands an observed Bell value $\eta$ exceedingly close to the absolute Tsirelson bound $\tau(\theta,\phi,\omega)$.

\begin{figure}[ht!] 
    \centering
    \includegraphics[width=0.9\linewidth]{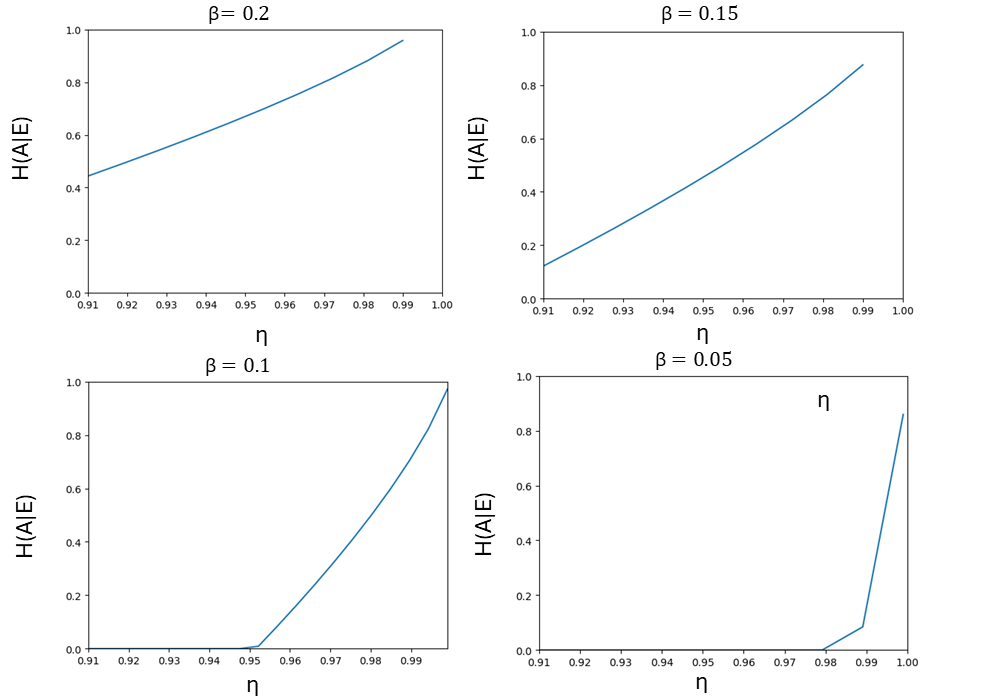}
    \caption{
    Plot of $H(A|\cE,X=0)$ versus $\eta$ for $\beta = 0.2,0.15,0.1,$ and $0.05$ for the Bell inequality used in \cite{Farkas_2024}. While the theoretical maximum key rate remains robust near ideal violations ($\cB_d \to 1$), smaller values of $\beta$ cause the curve to steepen sharply, demonstrating that noise tolerance shrinks rapidly as quantum correlations approach the classical boundary.
    }
    \label{fig:farkas}
\end{figure}

We now consider a recent work that obtains related results using a different family of Bell inequalities  parametrized by an integer \( d \ge 2 \) and an overlap matrix $O_{jk} = \vert{}\langle j\vert{}e_k\rangle\vert{}$ constructed from two orthonormal bases on $\mathbb{C}^d$ \cite{Farkas_2024}.

Alice has two measurement settings with \( d \) outcomes each and Bob has \( d^2 \) settings with 3 outcomes each. We denote the settings of Bob by the pair \( jk \), where \( j,k \in [d] \). For every \( d \ge 2 \) and every overlap matrix such that \( O_{jk} < 1 \) for all \( j,k \) Alice and Bob try to  violate the following a Bell inequality:
\begin{align}
\cB_d
= {} & \sum_{j,k=0}^{d-1} \sqrt{1 - O_{jk}^2}
\Bigl[
p(j,0 \mid 0,jk)
- p(j,1 \mid 0,jk)
\nonumber\\
& \qquad\qquad\qquad
+ p(k,1 \mid 1,jk)
- p(k,0 \mid 1,jk)
\Bigr]
\nonumber\\
& - \frac12
\sum_{j,k=0}^{d-1}
\left(1 - O_{jk}^2\right)
\Bigl[
p_B(0 \mid jk) + p_B(1 \mid jk)
\Bigr].
\end{align}
where \( p_B(b \mid jk) \) is Bob’s marginal distribution.

For this particular inequality, the Tsirelson bound is  obtained as $\cB_d^Q=d-1$, and it is guaranteed to be strictly larger than the classical maximum.
The maximal quantum violation of these inequalities provides a weak form of self-testing, which is sufficient to certify a device-independent key rate of $\log d$ \cite{Farkas_2024}.

By applying a continuous unitary perturbation $U_\beta$ to Alice's second measurement basis, the Tsirelson bound can be tuned arbitrarily close to the classical bound with the gap between classical and quantum bound vanishing  as $\beta \to 0^+$. Remarkably, even when the Tsirelson bound lies arbitrarily close to the classical limit, the underlying correlations satisfy self-testing conditions that guarantee an extractable device-independent key rate of $\log d$ bits.

In this Bell scenario, Alice has two measurement settings with $d$ outcomes each and Bob has $d^2$ settings with three outcomes each. Thus even for the smallest value of $d=2$, we cannot apply
Jordan's Lemma to simplify our analysis like we did for $\langle T_{\theta,\phi,\omega} \rangle$. However, we can approach this problem numerically as the bounds obtained on $H(A|\cE,X=0)$ from the techniques given in \cite{Brown_2024} are independent of the Jordan Lemma. We plot the results for this Bell inequality in Fig.~\ref{fig:farkas} for $\beta$ taking the values $0.2,0.15,0.1,$ and $0.05$.

The numerical results exhibit the exact same trend observed in Fig.~\ref{numericals: wooltorton}: as $\beta$ decreases and the quantum state approaches the local set, the key rate curve steepens dramatically. While larger values of $\beta$ (such as $\beta = 0.2$) display a forgiving, gradual decay of secret key rate under noise, smaller values ($\beta = 0.05$) drop off sharply, requiring the normalized Bell violation $\eta$ to be exceedingly close to $1$ to yield a positive key rate. This reinforces that while device-independent key distillation near the classical boundary remains theoretically achievable, the experimental margin for error vanishes as non-locality weakens.

\section{Discussion}
\label{sec:discussion}
We have shown that the key rates obtained by the maximal violation of the family of Bell inequalities introduced in \cite{Le_2023} are indeed robust. 
Particularly, we have shown that the strategies that achieve a slightly suboptimal violation give a non-zero key rate that has explicit dependence on the parameter $\epsilon$. We have also given numerical examples on the bounds of the conditional entropy for such suboptimal violation for two families of Bell inequalities. 

This sheds light on other recent results of DIQKD in the vicinity of the local set. It was shown that for convex-combination attacks, the correlations arising from the Werner state $\rho_v = v \ket{\psi^-}\!\bra{\psi^-} + (1-v)\frac{1}{4}\1$ with $v\leq v_{\text{crit}}$ give key rate zero for DIQKD protocols that allow two-way communication between Alice and Bob \cite{Farkas_2021}, but derive the key from the outputs of the nonlocal behaviour, while eventually publishing the inputs. A natural follow-up question is whether there are other similar non-local correlations close to the local set with zero key for any DIQKD protocol. Recent results have shown that this is not the case for DIQKD protocols achieving the Tsirelson bound for certain Bell inequalities \cite{Wooltorton_2024, Farkas_2024}. However, the answer was still open for correlation within the boundary of the quantum set, see Fig. \ref{fig:results}(a). In this region, while staying close to the local set, one might wonder whether there are points $P_2$ arbitrarily close to the extreme boundary where the key rate is zero. Our results prove that this is not the case and the key rates obtained in \cite{Wooltorton_2024,Farkas_2024} are robust, as shown in Fig.~\ref{fig:results}(b). 

\begin{figure}[ht]
\hspace*
{0.6cm}\includegraphics[width=\linewidth]{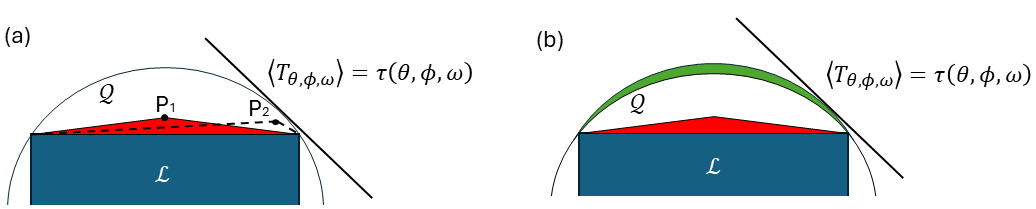}
\caption{Fig.~(a) illustrates the nonlocal region that gives zero key rate in the red triangle following the convex-combination attack from \cite{Farkas_2021}: $P_1$ corresponds to the behaviour generated by a  Werner state $\rho_v = v \ket{\psi^-}\!\bra{\psi^-} + (1-v)\frac{\1}{4}$ with $v\leq v_{crit}$.
One might ask whether there is another point $P_2$ in the quantum region arbitrarily close to the boundary, such that a convex-combination attack (denoted by the dotted black line) gives a zero key rate. In Fig.~(b) illustrates why this is not the case, as from Theorem \ref{thm:robust-key-rate} we obtain a non-zero key in the region highlighted in green, which from every extreme point identified by the Bell correlator $T_{\theta,\phi,\omega}$ extends some nonzero distance into the interior of the quantum region.}
\label{fig:results}
\end{figure}

The results of \cite{Wooltorton_2024,Farkas_2024} and the present ones shed light also on some older lines of thinking about no-signalling correlations, starting from Popescu and Rohrlich's \cite{PR-box} identification of no-signalling as the crucial backdrop against which to discuss quantum correlations: namely that in foundational matters and all the way to information processing, one could do away with the quantum objects and treat the observable behaviours themselves as a fundamental resource \cite{Barrett-resource}. The fact that arbitrarily close to the local set we have correlations with maximum, or near-maximum, DI secret key shows that this view has severe issues in that a basic information theoretic task exhibits an extreme discontinuity. Upper bounds on the DI secret key, such as the intrinsic nonlocality \cite{KaurWildeWinter:intrinsic} thus necessarily exhibit the same discontinuity (answering an open question from the latter paper), despite the fact that it is built from innocuous-looking entropic expressions. Extending the insights of \cite{Wooltorton_2024,Farkas_2024}, our results here reinforce that point by showing that the discontinuity is not restricted to the ideal boundary of the quantum set but rather survives in its interior. 

This work opens up several directions for future research. 
Our results are established in the asymptotic regime, which allows the problem to be reduced to the i.i.d. setting. Extending the analysis to the finite-key regime in a fully general scenario is a challenging open problem.
Finally, the present approach is restricted to the two-input–two-output scenario due to its reliance on Jordan's Lemma. An interesting direction for future work would be to investigate whether similar methods or else insights from the NPA hierarchy \cite{Pironio_2010} can be extended to more general Bell inequalities and measurement settings, which would provide a rigorous explanation of our numerical findings.

\section*{Acknowledgments} 
The authors thank Lewis Wooltorton for helpful conversations on robust versions of the result in \cite{Wooltorton_2024}, and Gereon Ko{\ss}mann for further discussions on robust self-testing during Beyond IID 2025 at IAS, Technische Universit\"at M\"unchen. 
They further thank Amerigo Bonasera and Salvatore Corsitto for emphasizing their belief in America. 
The authors are supported by the Spanish MICIN 
(project PID2022-141283NB-I00) with the support of FEDER funds; HKSV is additionally supported by the MICIN scholarship (FPI) PREP2022-000456. 
AW furthermore was partially supported by the Spanish MICIN with funding from European Union NextGenerationEU (PRTR-C17.I1) and the Generalitat de Catalunya; 
by the Spanish MTDFP through the QUANTUM ENIA project: Quantum Spain, funded by the European Union NextGenerationEU within the framework of the ``Digital Spain 2026 Agenda''; 
by the European Commission QuantERA project ExTRaQT (Spanish MICIN grant no.~PCI2022-132965);
and by the Alexander von Humboldt Foundation.


\section*{Author contribution statements.}Both the authors contributed equally to the research and writing. 

\printbibliography

\end{document}